\documentclass[12pt]{article}
\usepackage{amsmath}
\usepackage{graphicx,psfrag,epsf}
\usepackage{enumerate}
\usepackage{natbib}
\usepackage{url} 

\usepackage{algorithm}
\usepackage{algorithmic}
\usepackage{bm}
\usepackage{amssymb, latexsym}
\usepackage{amsthm}
\usepackage[colorlinks,linkcolor=red]{hyperref}
\usepackage{subfigure}

\newcommand{\blind}{1}

\newcommand{\nop}[1]{}

\newtheorem{theorem}{Theorem}[section]
\newtheorem{proposition}{Proposition}[section]
\newtheorem{lemma}{Lemma}[section]
\newtheorem{remark}{Remark}[section]
\newtheorem{definition}{Definition}[section]

\begin{document}

\def\spacingset#1{\renewcommand{\baselinestretch}%
{#1}\small\normalsize} \spacingset{1}


\if1\blind
{
  \title{\bf Community Detection by $L_0$-penalized Graph Laplacian}
  \author{Chong Chen\\
    School of Mathematical Science , Peking University\\
    and \\
    Ruibin Xi \thanks{
    Correspondence should be addressed to. Email: ruibinxi@math.pku.edu.cn. The authors gratefully acknowledge \textit{the National Natural Science Foundation of China (11471022, 71532001) and the Recruitment Program of Global Youth Experts of China}}\hspace{.2cm}\\
    School of Mathematical Science , Peking University\\
    and\\
    Nan Lin \\
    Department of Mathematics, Washington University in St. Louis}
  \maketitle
} \fi

\if0\blind
{
  \bigskip
  \bigskip
  \bigskip
  \begin{center}
    {\LARGE\bf Community Detection by $L_0$-penalized Graph Laplacian}
\end{center}
  \medskip
} \fi

\bigskip
\begin{abstract}

Community detection in network analysis aims at partitioning nodes in a network into $K$ disjoint communities. Most currently available algorithms assume that $K$ is known, but choosing a correct $K$ is generally very difficult for real networks. In addition, many real networks contain outlier nodes not belonging to any community, but currently very few algorithm can handle networks with outliers. In this paper, we propose a novel model free tightness criterion and an efficient algorithm to maximize this criterion for community detection. This tightness criterion is closely related with the graph Laplacian with $L_0$ penalty. Unlike most community detection methods, our method does not require a known $K$ and can properly detect communities in networks with outliers.

Both theoretical and numerical properties of the method are analyzed. The theoretical result guarantees that, under the degree corrected stochastic block model, even for networks with outliers, the maximizer of the tightness criterion can extract communities with small misclassification rates even when the number of communities grows to infinity as the network size grows. Simulation study shows that the proposed method can recover true communities more accurately than other methods. Applications to a college football data and a yeast protein-protein interaction data also reveal that the proposed method performs significantly better.

\end{abstract}

\noindent%
\vfill

\newpage
\spacingset{1.45} 
\section{Introduction}
\label{sec:intro}

Community detection has attracted tremendous research attention, initially in the physics and computer science community \citep{newman2004coauthorship,newman2004finding,newman2006modularity} and more recently in the statistics community \citep{bickel2009nonparametric,bickel2013asymptotic,zhao2012consistency,jin2015fast}.
Considering an undirected network $G=(V,E)$, where $V$ is the set of nodes and $E$ is the set of edges between nodes. Community detection is to find an ``optimal" partition of the nodes $V=G_1\bigcup \cdots \bigcup G_K$ such that nodes within the communities $G_k$ ($k=1,\cdots,K$) are more closely connected than nodes between the communities.

One class of community detection algorithms detects community by optimizing a heuristic global criterion over all possible partitions of the nodes \citep{wei1989towards,shi2000normalized}.
For example, the criterion modularity \citep{newman2004finding} has been very popular in community detection and fast algorithms for maximizing modularity \citep{newman2004fast} have been developed and widely used. The well-known spectral clustering algorithms \citep{jin2015fast,balakrishnan2011noise,chaudhuri2012spectral,
rohe2011spectral,joseph2016impact} can be traced back as continuous approximation methods of global criterion such as ratio cut \citep{hagen1992new} or modularity \citep{white2005spectral}. Spectral clustering methods are fast in computation and easy to implement since they usually only require calculation of a few eigenvectors of the Laplacian matrix.

Probabilistic model-based methods are another class of community detection algorithms. They detect communities by fitting a probabilistic model \citep{bickel2013asymptotic,nowicki2001estimation,mariadassou2010uncovering,decelle2011asymptotic} or by optimizing a criterion derived from a probabilistic model \citep{bickel2009nonparametric,karrer2011stochastic}. One of the most commonly used models is the stochastic block model (SBM) \citep{holland1983stochastic}. Given the adjacency matrix $A=(A_{ij})_{1\leq i,j\leq n}$ of a network $G$ with $n$ nodes, the SBM assumes that true node labels $c_i$ are independently sampled from a multinomial distribution with parameters $\mathbf{\pi}=\left(\pi_{1},...,\pi_{K}\right)^{T}$, i.e. $\pi_k = P(c_i=k),~k=1,\cdots, K$. Conditional on the community labels, the edges $A_{ij}$ ($i<j$) are independent Bernoulli random variables with $P(A_{ij}=1|c_i,c_j)=p_{c_{i}c_{j}}$. The connection probabilities should have $p^{-}=\min_k\{p_{kk}\}>q^{+}=\max_{k\neq m}\{p_{km}\}$. The SBM assumes that the expected degrees are the same for all nodes and does not allow hubs in networks. To remove this constraint, the degree corrected stochastic block model (DCSBM) \citep{karrer2011stochastic} introduces a degree correction variable $\theta_i$ to each node such that $P(A_{ij}=1|c_i,c_j,\theta_i,\theta_j)=\theta_i\theta_j p_{c_{i}c_{j}}$, where $\theta_i>0$ and $E(\theta_i)=1$.
\nop{Other models include the mixed membership models \citep{newman2007mixture,airoldi2008mixed} and latent models \citep{bickel2009nonparametric,hoff2002latent,karrer2011stochastic}.}
\medskip

Consistency results were developed for a number of community detection algorithms, mostly based on the SBM or DCSBM. Under the assumption that the community number is fixed, Bickel and Chen \citep{bickel2009nonparametric} laid out a general theory under the SBM for checking consistency of community detection criteria when the network size grows to infinity, and similar theories were also developed for DCSBM \citep{zhao2012consistency,jin2015fast}. With a fixed community number, the community size would linearly grow as the number of nodes grows. However, this is not a realistic assumption, because real networks often have tight communities at small scales, even when networks contain millions of nodes \citep{leskovec2008statistical}. Recent researches \citep{rohe2011spectral,choi2012stochastic,cai2015robust} generalized these consistency results by allowing the number of communities grows as the node number grows for the SBM. However, as far as we know, similar results for the DCSBM are not available yet.
\medskip

Despite all these progresses in community detection, most of the current algorithms assume that the number of communities $K$ is known in priori. This is problematic because the number of communities is usually unknown in real applications. If the network is small, one may try a few different community numbers $K$ and choose the ``best" $K$ as the ``true" community number. When the network is large, it would be very difficult to perform such a search. Although there are a few algorithms \citep{newman2004finding,zhao2011community} not requiring a known $K$, consistency results for these algorithms were either not developed or only developed in very simple cases (e.g. for SBMs with $K=2$). In addition, real networks often contain outlier nodes that cannot be grouped into any communities \citep{kumar2010structure,wang2015community,khorasgani2010top}. Currently, there is not much work for community detection in networks with outliers. A few algorithms were developed \citep{zhao2011community,lancichinetti2011finding}, but their performance is limited and, as far as we know, no consistency result is available for these algorithms yet.

\medskip

In this paper, we propose a novel model-free tightness criterion for community detection. Community detection based on this criterion iteratively extracts single communities and no prior knowledge about $K$ is needed. A permutation-based test is performed to filter the extracted communities that are likely to be outliers or false communities. Under the DCSBM, we establish asymptotic consistency allowing the community number $K$ increases as the number of nodes grows. We further extend this consistency result for DCSBMs with outliers. We show that maximizing this criterion is equivalent to maximizing a penalized graph Laplacian with constraints. An efficient algorithm is developed based on the alternating direction method of multiplier (ADMM) to maximize this penalized Laplacian. This paper is organized as follows. The model-free criterion and the ADMM algorithm are described in Section \ref{sec:meth}. Theoretical results are given in Section\ref{sec:theory}. Section \ref{sec:simulation} presents simulation comparison with existing methods and Section \ref{sec:realData} is the real data analysis. Proofs of the theorems are given in the Appendix.

\section{Method and Algorithm}
\label{sec:meth}
Assume that nodes of a graph $G=(V,E)$ are indexed by $\{1,2,...,n\}$ and each node $i$ belongs to exactly one of $K$ non-overlapping communities denoted by a latent label $c_{i}\in\{1,...,K\}$. If there is no outlier, nodes within communities are more tightly connected than nodes between communities. For networks with outliers, we assume that the $K$th community is the outlier ``community", in which nodes are randomly connected with other nodes in the network. Exact definition of outliers is given in Section \ref{sec:theory}. Given a set $S\subset V$, the complementary set of $S$ is denoted by $\bar{S}$ and the number of elements in $S$ is denoted as $|S|$. Define $W(S)=\sum_{i,j\in S}A_{ij}$, $B(S)=\sum_{i\in S,j\in \bar{S}}A_{ij}$ and $V(S) = W(S)+B(S)$. Then, $W(S)$ is twice the number of edges between nodes in $S$, $B(S)$ is the total number of edges between $S$ and $\bar{S}$ and $V(S)$ is the total degrees in $S$. Given a vector $\mathbf{u}$, we denote $\|\mathbf{u}\|_{0}$ as the number of nonzero elements in $\mathbf{u}$ and $\|\mathbf{u}\|_{2}$ as the $L_2$-norm of the vector $\mathbf{u}$.

\subsection{A tightness criterion}

Given a set $S\subset V$, if it is a true community, we expect that most of its connections are within $S$ itself and thus $W(S)/V(S)$ should be large. However, directly maximizing $W(S)/V(S)$ has a trivial solution $S=V$. We instead introduce a penalty to the size of the community and consider the following tightness criterion,
\begin{equation}
\label{Objective Function}
\psi(S)=\frac{W(S)}{V(S)}-\eta \left|S\right|,
\end{equation}
where $\eta$ is a tuning parameter. In Section 3, we will show that with a proper choice of $\eta$, maximizing this tightness criterion can render consistency in community detection. This tightness criterion is closely related to a penalized graph Laplacian. More specifically, let $Q=D^{-1/2}AD^{-1/2}$ be the graph Laplacian, where $A$ is the adjacency matrix and $D=\mbox{diag}\{d_1,\cdots,d_n\}$ is the nodal degree matrix with $d_i$ being the degree of the $i$th node. Then, we have the following proposition.

\begin{proposition}
	\label{thm:3.1}
	Given a set $S\subset V$, define its membership vector by
	\begin{equation}\label{communityTovector}
	\mathbf{u}_S(i)=\begin{cases}
	\frac{\sqrt{d_{i}}}{\sqrt{W(S)+B(S)}}, & \mbox{ if }\ \ i\in S, \\ 0, & \mbox{ if }\ \ i\in \bar{S}.
	\end{cases}
	\end{equation}
	
	Then we have $\psi(S)=\mathbf{u}_S^{t}Q\mathbf{u}_S - \eta\|\mathbf{u}_S\|_0$ and $\|\mathbf{u}_S\|_{2}=1$.
\end{proposition}
Therefore, maximizing the tightness criterion (\ref{Objective Function}) is equivalent to the following optimization problem
\begin{equation}
\label{obj_equivalent}
\max\limits_{S\subset V, \mathbf{u}=\mathbf{u}_S}\ \mathbf{u}^{t}Q\mathbf{u}-\eta\|\mathbf{u}\|_{0}.
\end{equation}
The objective function in (\ref{obj_equivalent}) is the penalized graph Laplacian and hence maximizing (\ref{Objective Function}) is equivalent to maximizing the penalized graph Laplacian with the contraints $\mathbf{u}=\mathbf{u}_S$, $S\subset V$. Finding global solution to (\ref{obj_equivalent}) is difficult in general, because we have to search over all possible subsets of $V$ to find the maximum. In the next section, we develop an efficient greedy algorithm based on the ADMM to solve (\ref{obj_equivalent}).

\subsection{Algorithm}

Before giving the algorithm, we first introduce a few notations. For any $\mathbf{u}$ with $\|\mathbf{u}\|_2=1$, we denote its nonzero element index set $S(\mathbf{u})=\{i:~\mathbf{u}(i)\neq 0\}\subset V$. On the other hand, given $S(\mathbf{u})$, we can define a new vector $\mathbf{u}_d=\mathbf{u}_{S(\mathbf{u})}$ using (\ref{communityTovector}). Thus, we have $\|\mathbf{u}_d\|_2=1$. The vector $\mathbf{u}_d$ is obtained by just reassigning values of the nonzero elements of $\mathbf{u}$ according to the degrees of $S(\mathbf{u})$. Given $\lambda_1\geq 0$, we consider the following optimization problem
\begin{equation}\label{obj_equivalent_augment}
\max\limits_{\|\mathbf{u}\|_{2}=1}\ \mathbf{u}^{t}Q\mathbf{u}-\eta\|\mathbf{u}\|_{0}-2\lambda_1\|\mathbf{u}-\mathbf{u}_d\|_2^2.
\end{equation}
The optimization problem (\ref{obj_equivalent_augment}) can be viewed as the augmented Lagrangian of (\ref{obj_equivalent}). If we take $\lambda_1$ as $\infty$, $\mathbf{u}$ in (\ref{obj_equivalent_augment}) will be forced to be $\mathbf{u}_d$. In (\ref{obj_equivalent}), the constraint is that $\mathbf{u}$ can only take discrete values; but $\mathbf{u}$ in (\ref{obj_equivalent_augment}) can be any vector with norm 1 and hence solving (\ref{obj_equivalent_augment}) is much easier than solving (\ref{obj_equivalent}). By introducing an intermediate variable $\mathbf{v}$ with $\mathbf{v}=\mathbf{u}$, the augmented Lagrangian of (\ref{obj_equivalent_augment}) is
\begin{equation}
\label{equ:3.8}
\max\limits_{\|\mathbf{u}\|_{2}=1,\|\mathbf{v}\|_{2}=1}\mathbf{u}^{t}Q\mathbf{v}-\lambda\|\mathbf{u}-\mathbf{v}\|_{2}^2-\frac{\eta}{2}(\|\mathbf{u}\|_{0}+\|\mathbf{v}\|_{0})-\lambda_1\|\mathbf{u}-\mathbf{u}_d\|_2^2-\lambda_1\|\mathbf{v}-\mathbf{v}_d\|_2^2,
\end{equation}
which is equivalent to
\begin{equation}
\label{equ:3.9}
\max\limits_{\|\mathbf{u}\|_{2}=1,\|\mathbf{v}\|_{2}=1}\mathbf{u}^{t}(Q+2\lambda I)\mathbf{v}-\frac{\eta}{2}(\|\mathbf{u}\|_{0}+\|\mathbf{v}\|_{0})+2\lambda_{1}\mathbf{u}^{t}\mathbf{u}_{d}+2\lambda_{1}\mathbf{v}^{t}\mathbf{v}_{d}.
\end{equation}
We can solve (\ref{equ:3.9}) by iteratively updating $\mathbf{u}$ and $\mathbf{v}$. When either $\mathbf{u}$ or $\mathbf{v}$ is given, problem (\ref{equ:3.9}) reduces to a simple linear programming problem with an explicit solution given in the following proposition. The proof of this proposition is given in \citet{kim2012scalable} and we omit it here.
\begin{proposition}
	\label{thm:3.4}
	For a given vector $\mathbf{z}=(z_{1},...,z_{n})^{t}\in \mathbb{R}^{n}$, we denote its $r$th largest absolute value  as $\left|z\right|_{r}$, and let $\mathbf{z}_{r}^{h}$ be the vector with the $i$th element as $\mathbf{z}_{r}^{h}(i)=z_{i}\mathbf{I}(\left|z_{i}\right|>\left|z\right|_{r+1})$. Then for a constant $\rho>0$, the solution to
	\begin{equation}
	\label{equ:3.6}
	\max\limits_{\|\mathbf{u}\|_{2}=1}\mathbf{u}^{t}\mathbf{z}-\rho \|\mathbf{u}\|_{0}
	\end{equation}
	is $\mathbf{u}=L(\mathbf{z},\rho)=\mathbf{z}_{r}^{h}/\|\mathbf{z}_{r}^{h}\|_{2},$
	where $r$ is the smallest integer that satisfies
	\begin{equation}
	\label{equ:3.7}
	\left|z\right|_{r+1}\leq \sqrt{\rho^{2}+2\rho\|\mathbf{z}_{r}^{h}\|_{2}}.
	\end{equation}
\end{proposition}
We summarize the algorithm for solving (\ref{equ:3.9}) as Algorithm \ref{alg:1} below
\begin{algorithm}[H]
	\caption{$L_{0}$-Penalized Laplacian Algorithm }
	\label{alg:1}
	\begin{algorithmic}
		\REQUIRE $Q,~\lambda,~\lambda_{1},~\eta$ and $\epsilon$\\
		\STATE{Initialize $\mathbf{v}^0,~\mathbf{u}^0$. For each $k=1,2,\cdots,$}
		\REPEAT
		\STATE{$\mathbf{z}_1^k= (Q+2\lambda I)\mathbf{v}^{k-1}+2\lambda_{1}\mathbf{u}^{k-1}_{d},$ $\mathbf{u}^k=L(\mathbf{z}_1^k,\eta/2)$;}
		\STATE{$\mathbf{z}_2^k= (Q+2\lambda I)^{t}\mathbf{u}^k+2\lambda_{1}\mathbf{v}^{k-1}_{d},$ $\mathbf{v}^k=L(\mathbf{z}_2^k,\eta/2)$;}
		\UNTIL{$\|\mathbf{u}^k-\mathbf{v}^k\| < \epsilon$.}
		\RETURN{$S_\eta=\{i,\mathbf{u}^k(i)\neq 0~\mbox{and}~\mathbf{v}^k(i)\neq0\}$}
	\end{algorithmic}
\end{algorithm}
\nop{
Particularly, if we take $\lambda_1=0$, Algorithm \ref{alg:1} can be reduced to Algorithm \ref{alg:2} below.
\begin{algorithm}[H]
	\caption{$L_{0}$-Penalized Laplacian Algorithm }
	\label{alg:2}
	\begin{algorithmic}
		\REQUIRE $Q,~\lambda,~\eta$ and $\epsilon$\\
		\STATE{Initialize $\mathbf{v}^0$. For each $k=1,2,\cdots,$}
		\REPEAT
		\STATE{$\mathbf{z}_1^k= (Q+2\lambda I)\mathbf{v}^{k-1},$ $\mathbf{u}^k=L(\mathbf{z}_1^k,\eta/2)$;}
		\STATE{$\mathbf{z}_2^k= (Q+2\lambda I)^{t}\mathbf{u}^k,$ $\mathbf{v}^k=L(\mathbf{z}_2^k,\eta/2)$;}
		\UNTIL{$\|\mathbf{u}^k-\mathbf{v}^k\| < \epsilon$.}
		\RETURN{$S_\eta=\{i,\mathbf{u}^k(i)\neq 0~\mbox{and}~\mathbf{v}^k(i)\neq0\}$}
	\end{algorithmic}
\end{algorithm}
}
In all simulation and real data analysis, we fix $\lambda$ as ${1}/{\sqrt{n}}$  and $\epsilon$ as $10^{-4}$ for Algorithm \ref{alg:1}.We find that Algorithm \ref{alg:1} is sensitive to the choice of initial values. In order to get a more robust result, we first run Algorithm \ref{alg:1} with the initial value $\mathbf{v}^0 = \left(1/\sqrt{n},...,1/\sqrt{n}\right)$ and $\lambda_1=0$ to get $\hat{\mathbf{v}}^0$. Then we run Algorithm \ref{alg:1} with the initial value $\hat{\mathbf{v}}^0$ and $\lambda_1=1$ to get the final solution. We call Algorithm \ref{alg:1} $L_{0}$-Penalized Laplacian Algorithm (L0Lap). We discuss how to tune the parameter $\eta$ in the following subsection. Algorithm \ref{alg:1} extracts one community at a time from the network. After the first community is identified, we iteratively apply Algorithm \ref{alg:1} to the remaining network until there is no edge left in the network. This iterative procedure can detect true communities but also may generate false communities. For example, in a SBM with small connecting probabilities, a few nodes (e.g. 2 nodes) could easily only connect to themselves but not connect to other nodes in the network. The iterative application of Algorithm \ref{alg:1} will capture these nodes as a community. However, such communities are most likely to be spurious, because even Erd\"{o}s-R\'{e}nyi (ER) networks can have such small communities. In the next subsection, we present a permutation test that can effectively filter false communities.

\subsection{Choice of the tuning parameter $\eta$ and the permutation test}

Given a subset $S\subset V$, define $\bar{p}_{W}(S)=W(S)/(|S|(|S|-1))$ and  $\bar{p}_{B}(S)=B(S)/(|S||\bar{S}|)$. Thus, $\bar{p}_{W}(S)$ is the average connection within $S$ and $\bar{p}_{B}(S)$ is the average connection between $S$ and $\bar{S}$. If $S$ is the $k$th community of a SBM, $\bar{p}_{W}(S)$ would be approximately the same as the connection probability $p_{kk}$. If the SBM has a constant between-community connection probability $q^+$ (i.e. $p_{km}=q^+$ for all $k\neq m$), $\bar{p}_{B}(S)$ would be approximately the same as $q^+$. Therefore, $\bar{p}_{W}(S)$ would be much larger than $\bar{p}_{B}(S)$ (assuming $p_{kk}$ much larger than $q^+$). On the other hand, if $S$ consists of nodes from different communities, then $\bar{p}_{W}(S)$ would be generally smaller than $p^{-}$. Define
\begin{equation}
\label{equ:3.1}
\phi(S)=\frac{\bar{p}_{W}(S)}{\bar{p}_{W}(S)+\bar{p}_{B}(S)}.
\end{equation}
True communities should have large $\phi(S)$ compared with random sets of nodes. In this paper, we use $\phi(S)$ to help choosing the tuning parameter $\eta$. In comparison, \citet{zhao2011community} used a criterion, the difference between $\bar{p}_{W}(S)$ and $\bar{p}_{B}(S)$, that is similar to $\phi(S)$ for community detection. From the theoretical results in the next section, we know that the best $\eta$ is at the order of $O(1/n)$. Therefore, we run Algorithm \ref{alg:1} for $\eta=0,B/n,2B/n,\cdots,cB/n$ and choose an $\eta$ such that the resulted $S_\eta$ from Algorithm \ref{alg:1} corresponds to the maximum value of $\phi(S_\eta)$. In all simulation and real data analysis, we choose $c=10$ and $B=1$.

After iterative application of Algorithm \ref{alg:1} to the network, we further use a permutation test to filter false communities. Suppose that $S_1,\cdots,S_c$ are all the identified communities with less than $M$ nodes and $G_{0}$ is the sub-network of $G$ composed of nodes in $\bigcup_{i=1}^cS_i$. If $M$ is large enough, $G_0$ will be the same as $G$. To save computational time, we choose $M = 20$ and apply a permutation test to $G_{0}$. Let $\bar{p}=\sum_{i,j}A_{ij}/(n^2-n)$. Given a sub-network $S$ of $G_0$, if $S$ is an ER-graph with a connecting probability $\bar{p}$, given any $m$ nodes, the probability of observing no more than $E$ edges between these $m$ nodes is
\begin{equation}\label{permutationProbEq}
p(m,E)=\sum_{i=0}^E{m(m-1)/2 \choose i}\bar{p}^i(1-\bar{p})^{m(m-1)/2-i}.
\end{equation}

Let $n_i$ and $E_i$ be the number of nodes and number of edges in $S_i$ ($i=1,\cdots,c$), respectively. Each detected community $S_i$ has an associated probability $p(n_i,E_i)$ using (\ref{permutationProbEq}). We permute $N$ times the edges in $G_0$ to generate $N$ ER-graphs and run Algorithm \ref{alg:1} to each of the $N$ ER-graphs. The first extracted community of the $j$th ER-graph also has a probability $p_{j}^{ER}$ using (\ref{permutationProbEq}). We assign the permutation p-value for $S_i$ as $p_i=|\{j:~p_{j}^{ER}\leq p(n_i,E_i),~j=1,\cdots,N\}|/N$. The detected community $S_i$ is filtered out if $p_i\geq \alpha$. In our simulation and real data, we set $N=100$ and $\alpha=0.05$.

\section{Theoretical Properties}
\label{sec:theory}

In this section, we discuss theoretical results about the estimator $S$ that maximizes the tightness criterion (\ref{Objective Function}) under the DCSBM. We first give the exact definition of the DCSBM.
\begin{definition}
\label{dcsbm}
A network $G=(V,E)$ is said to follow a DCSBM, if it satisfies the following assumptions.
\begin{enumerate}
	\item[(A1)] Each node is independently assigned a pair of latent variables $(c_{i},\theta_{i})$, where $c_{i}$ is the community label taking values in $\{1,2,...,K\}$, and $\theta_{i}$ is a ``degree variable" taking discrete values in $\{h_1,\cdots,h_M\}$ ($0<h_1<...<h_M$).
	\item[(A2)] The marginal distribution of $c_i$ is a multinomial distribution with parameters $\bm\pi=\left(\pi_{1},...,\pi_{K}\right)^{T}$, and the random variable $\theta_i$ satisfies $E[\theta_{i}]=1$ for identifiability.
	\item[(A3)] Given $\mathbf{c}=(c_{1},...,c_{n})$ and $\bm\theta=(\theta_{1},...,\theta_{n})$, the edges $A_{ij}$ ($i<j$) are independent Bernoulli random variables with $P(A_{ij}=1|\mathbf{c},\bm\theta)=\theta_{i}\theta_{j}p_{c_{i}c_{j}}$.
	\item[(A4)] Denote $\pi^- = \min_{1\leq k\leq K}\pi_{k}$, $p^{-}=\min_{1\leq k\leq K}\{p_{kk}\}$ and $q^{+}=\max_{k\neq m}\{p_{km}\}$. Then, $p^{-}>q^{+}$.

\end{enumerate}
\end{definition}
Throughout this paper, we assume that $\alpha, \tau, \gamma, \delta$ are fixed constants such that $0\leq 2\delta<\alpha<1/2$ and $0<\tau<\gamma<\alpha-2\delta$. These constants are for controlling the community number $K$, the connecting probabilities $p^-$ and $q^+$ and the smallest community size parameter $\pi^-$.  Given any two sets $S_{1}$ and $S_{2}$, we denote $S_{1}\Delta S_{2}=S_{1}\bigcup S_{2}-S_{1}\bigcap S_{2}$ as their symmetric difference. For two nonnegative sequences $a_n$ and $b_n$, we write $a_n\gtrsim b_n$ if there exists a constant $C_{0}>0$ such that $a_n\geq C_0b_n$. We assume $p^{-}\gtrsim {\log n}/{n^{1-2\alpha}}$ and $\pi^{-}\gtrsim n^{-\delta/2}$. Define $\Gamma_{\delta}=\{S\subset V , {|S|^{2}}/{K}\gtrsim n^{2-2\delta}\}$. Similar to \citet{zhao2012consistency}, we assume $\Pi$ is the $K\times M$ matrix representing the joint distribution of $(c_{i},\theta_{i})$ with $\mathbb{P}(c_{i}=k,\theta_{i}=h_{l})=\Pi_{kl}$. Denote  $\pi_{k}^{d}=\sum_{l}h_{l}\Pi_{kl}$. Note that since $E[\theta_{i}]=1$, we have $\sum_{k}\pi_{k}^{d}=1$. Let $\rho_{k}^{d}={p_{kk}}/{\sum_{l=1}^{K}\pi_{l}^{d}p_{kl}}$.

Given the community label $\mathbf{c}$ and a set of nodes $S\in \Gamma_{\delta}$, denote $G_{k}=\{i|c_{i}=k,~i=1,\cdots,n\}$, $S_{k}=\{i|i\in S, c_{i}=k\}$, $\hat{\pi}_{k}={|G_{k}|}/{n}$ and $r_{k}(S)={|S_{k}|}/{n}$ for $1\leq k\leq K$. We define $\hat{\pi}_{k}^{d}={\pi_{k}^{d}}\hat{\pi}_{k}/\pi_{k}$, $r_{k}^{d}(S)={\pi_{k}^{d}}r_{k}(S)/\pi_{k}$, $r^{d}(S)=\sum_{k=1}^{K}r_{k}^{d}(S)$ and $\hat{\rho}_{k}^{d}={p_{kk}}/{\sum_{l=1}^{K}\hat{\pi}_{l}^{d}p_{kl}}$. For $S=G_1$, we have $r_1(G_1)=\hat{\pi}_1$, $r_{1}^{d}(G_1)=\hat{\pi}_1^d$, $r_k(G_1)=0$, $r_{k}^{d}(G_1)=0$ ($k=2,\cdots,K$) and $r^{d}(G_1)=\hat{\pi}_1^d$. Let $x_{k}=p_{kk}$, $y_{k}=\sum_{l=1}^{K}\hat{\pi}_{l}^{d}p_{kl}$ for $1\leq k\leq K$. For any $t_k\geq 0$ ($k=1,\cdots,K$) and $\sum_{k=1}^Kt_k=1$, define $$f(t_{1},...,t_{K})=\frac{\sum_{k=1}^{K}t_{k}(t_{k}x_{k}+\sum_{l\neq k}^{K}t_{l}p_{kl})}{\sum_{k=1}^{K}t_{k}y_{k}}.$$

\begin{theorem}
	\label{DCSBM}
	Assume $\rho_{1}^{d}-\max_{2\leq k\leq K}\rho_{k}^{d}\gtrsim n^{-\tau}$ and ${\pi_{1}^{d}}/{\pi_{1}}\geq\max_{2\leq k\leq K}{\pi_{k}^{d}}/{\pi_{k}}$. Then, there is a constant $C$ such that, with probability at least $1-2Kn^{-2}$, we can choose $\eta>0$ such that
	\begin{align}\label{etacondition}
	f(1,0,...,0)\frac{\pi_{1}^{d}}{\pi_{1}}-\frac{C}{n^{\gamma}}>n\eta>\max_{t_{1}\leq 1-1/n^{\gamma-\tau}}f(t_{1},t_{2},...,t_{K})\frac{\pi_{1}^{d}}{\pi_{1}}+\frac{C}{n^{\gamma}}.
	\end{align}
With such a choice of $\eta$, suppose that $S\subset V$ is such that the tightness criterion (\ref{Objective Function}) is maximized in $\Gamma_{\delta}$, then with probability at least $1-(2K)n^{-2}-2^{n+2}/n^n$,
	\begin{equation}
    \label{convergerate}
	\frac{\left|S\Delta G_{1}\right|}{\left|S\bigcup G_{1}\right|}\leq 2h_{M}h_{1}^{-1}/n^{\gamma-\tau}+\log n/n^{\alpha-2\delta-\gamma}.
	\end{equation}
\end{theorem}

This theorem says that under a number of regularity conditions, if the tuning parameter is chosen properly, the detected community $S$ is very close to the underlying true community $G_1$. The condition ${\pi_{1}^{d}}/{\pi_{1}}\geq\max_{2\leq k\leq K}{\pi_{k}^{d}}/{\pi_{k}}$ is not as restrictive as it looks. For example, if $\mathbf{c}$ and $\bm\theta$ are independent, then $\pi_{k}^{d}=\pi_{k}$ for all $1\leq k\leq K$ and this condition is naturally satisfied. The SBM clearly also satisfies this condition, since in this case $M=1$ and $h_1=1$.  The first extracted community $G_1$ corresponds to the community with the largest $\rho_k^{d}$. For the SBM, when $p_{kl}=p_{0}$ for all $k\neq l$, we have $\rho_{k}^{d}=p_{kk}/\left(p_{kk}\pi_{k}+p_{0}\left(1-\pi_{k}\right)\right)=1/\left(\left(1-p_{0}/p_{kk}\right)\pi_{k}+p_{0}/p_{kk}\right)$. The ratio $\beta_{k}=p_{0}/p_{kk}$ can be viewed as the ``out-in-ratio" defined in \citep{decelle2011asymptotic}, and we have $\rho_{k}^{d}=1/\left(\left(1-\beta_{k}\right)\pi_{k}+\beta_{k}\right)$. If all $\pi_i$'s are the same,  the first extracted community $G_1$ is the community with the smallest out-in-ratio. If all out-in-ratios $\beta_{k}$ are the same, the first extracted community $G_1$ is the community with the smallest size.

Since $p^{-}\gtrsim {\log n}/{n^{1-2\alpha}}$ and $\pi^{-}\gtrsim n^{-\delta/2}$, we have $p^{-}n^{1-2\alpha+\delta/2}\gtrsim K\log n$. Consider a special case when $K$ is finite and the community sizes are all $O(n)$. In this case, the lower bound of the connecting probability within communities should satisfies $p^{-}\gtrsim {\log n}/{n^{1-2\alpha+\delta/2}}$ and thus $np^{-}/\log n\gtrsim n^{2\alpha-\delta/2}$. This condition is similar to the condition $np^{-}/\log n\rightarrow \infty$ in \citet{zhao2012consistency}, especially when $\alpha$ is close to 0. If $p^{-}=O(1)$ and $\delta = 1/4 - 2\epsilon$ for some $\epsilon>0$ very close to 0, then $n_{min} = O(n^{7/8+\epsilon})$ and $K=O(n^{1/8-\epsilon})$. Thus, the upper bound of $K$ is $O(n^{1/8})$. \nop{In fact, the relationship between $p^{-}$ and $K$ implies the expected number of edges in each community should be big enough. In other word, each community contains enough signals. This relationship had also been discussed in \citet{cai2015robust}.} Consider the simplest case when $K=2$. Let $\tau=0$ and $\gamma=\alpha/2-\delta$, the misclassification rate is about $O_{p}(\log n/n^{\alpha/2-\delta})$ by the inequality (\ref{convergerate}). This improves the results in \citet{rohe2011spectral} and \citet{lei2015consistency} where the misclassification rate was  $O_{p}(1/\log n)$.




When there are outliers in networks, we also have a consistency result similar to Theorem \ref{DCSBM}. We first give the definition of the DCSBM with outliers.
\begin{definition}
\label{odcsbm}
A network $G=(V,E)$ is said to follow a DCSBM with outliers, if it satisfies all assumptions (A1)-(A3) and the following assumption.
\begin{enumerate}
	\item[(A4$^\prime$)] Denote $\pi^- = \min_{1\leq k\leq K-1}\pi_{k}$,
$p^{-}=\min_{1\leq k\leq K-1}\{p_{kk}\}$ and $q^{+}=\max_{k\neq m}\{p_{km}\}$. Then, $p^{-} > q^{+}\geq p_{KK}$. The $K$th community is called the outlier community.
\end{enumerate}
\end{definition}
For a DCSBM with outliers, all communities are well-defined communities except the $K$th outlier community. We also assume $p^{-}\gtrsim {\log n}/{n^{1-2\alpha}}$ and $\pi^{-}\gtrsim n^{-\delta/2}$ for the DCSBM with outliers. We have the following theorem.
\nop{\begin{remark}
	Assume $G_{K}$ is the set contain all outliers and each node in $G_{K}$ connects with other nodes with probability $p_{0}$. For any $S\subset G_{K}$, If $1-p_{0}=2^{-1/(|S|(n-|S|))}$, we have $\mathbb{P}\left(V(S)-W(S)=0\right)=(1-p_{0})^{|S|(n-|S|)}=1/2$ and $\mathbb{P}\left(W(S)>0\right)=1-(1-p_{0})^{|S|^{2}}=2^{-|S|/(n-|S|)}$. If $|S|=O(n)$, we have $\mathbb{P}\left(W(S)>0\right)=O(1)$, which implies $\mathbb{P}\left(W(S)/V(S)=1\right)=O(1)$. So $S$ will be treated as a community with high probability. Therefore, Theorem \ref{DCSBM} does not work when $|G_{K}|=O(n)$.
\end{remark}
}

\begin{theorem}
	\label{DCSBM_outlier}
Suppose that $G$ is a DCSBM with outliers. Assumes that all conditions in Theorem \ref{DCSBM} hold. In addition, assume that the outlier community $G_K$ satisfies $|G_K|^2/K=o(n^{2-2\delta})$. Then, the conclusions in Theorem \ref{DCSBM} hold.
\end{theorem}
This theorem says that as long as the outlier community is not too large, the first extracted community will be very close to the community with the largest $\rho^d_k$.

\section{Simulation Study}
\label{sec:simulation}
In this section, we perform simulation study to compare the methods proposed in this paper with currently available methods. For the algorithm developed in this paper, we consider two versions of the algorithm, with or without the permutation test. This helps us to see the effect of the permutation test on removing false communities. We call these algorithms L0Lap (without the permutation test), L0LapT (with the permutation test). The methods that we compare with include Newman's modularity \citep{newman2004fast}, SCORE \citep{jin2015fast}, oPCA \citep{chung1997spectral}, nPCA \citep{shi2000normalized} and OSLOM \citep{lancichinetti2011finding}. For Newman's modularity, we use the C++ implementation available at \url{http://cs.unm.edu/~aaron/research/fastmodularity.htm}. For OSLOM, we use the C++ implementation from \url{http://www.oslom.org/software.htm}. For the other three methods, we implement the algorithms using Matlab according to their respective descriptions. Since SCORE, oPCA and nPCA requires a known community number, we provide the true community number to these algorithms in the simulation. The simulation study includes both networks without outliers and networks with outliers.

\nop{
\begin{enumerate}[(1)]
	\item Newman [\textbf{RX: need a reference here; Is this the modularity method?}]: We use the C++ implementation available at\\
	\url{http://cs.unm.edu/~aaron/research/fastmodularity.htm}.
	\item SCORE:  We use a Matlab implementation of the SCORE algorithm according to \citet{jin2015fast}.
	\item oPCA: We use the Matlab implementation of the ordinary spectral clustering algorithm according to \citet{chung1997spectral}.
	\item nPCA: We use the Matlab implementation of the normalized spectral clustering algorithm according to \citet{shi2000normalized}.
	\item OSLOM: We use the C++ implementation available at \\
	\url{http://www.oslom.org/software.htm}.
\end{enumerate}
}

\begin{figure}[H]
	\begin{minipage}[t]{0.5\linewidth}
		\centering
		\includegraphics[width=0.9\textwidth,natwidth=510,natheight=542]{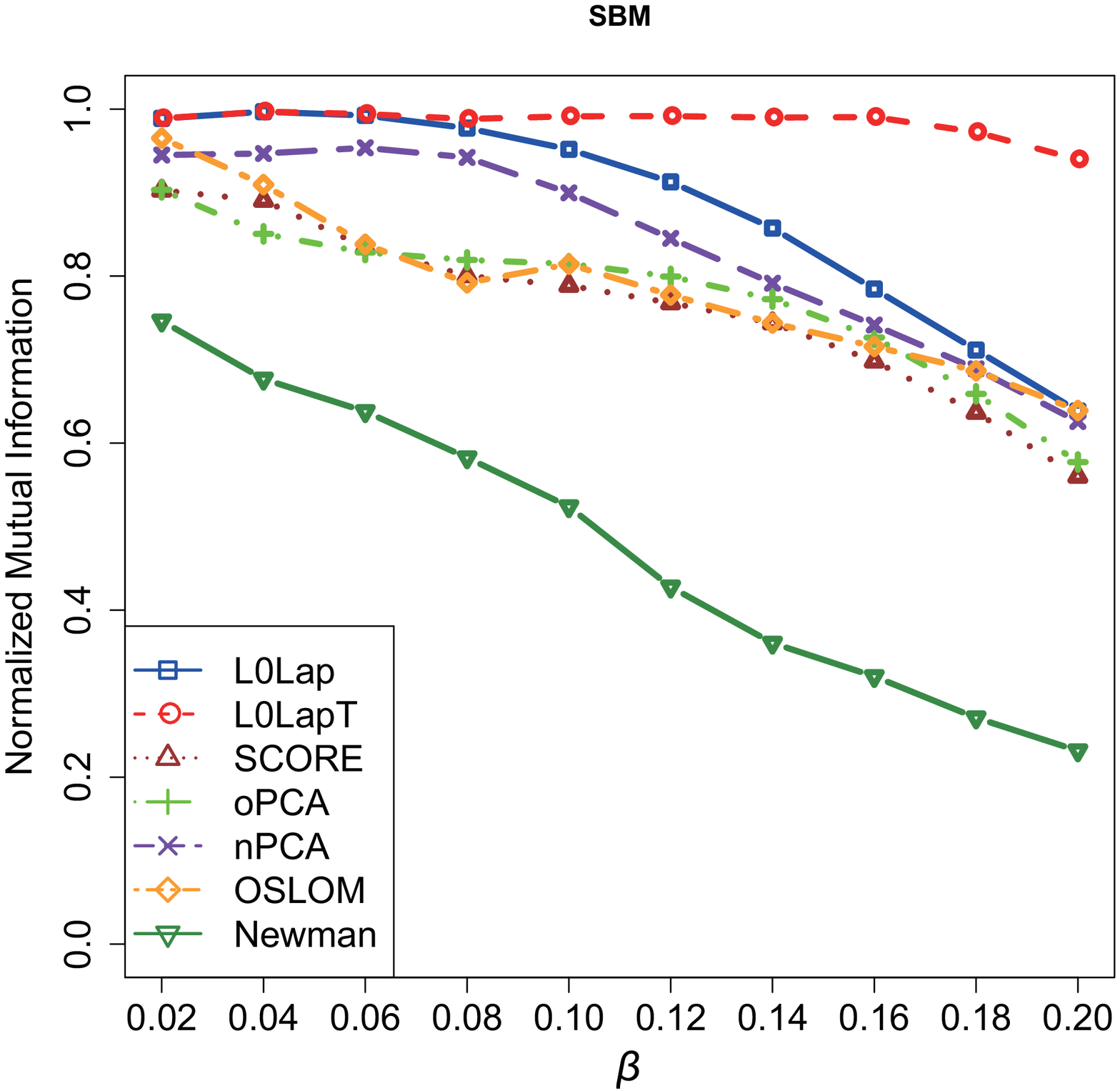}
	\end{minipage}
	\begin{minipage}[t]{0.5\linewidth}
		\centering
		\includegraphics[width=0.9\textwidth,natwidth=510,natheight=542]{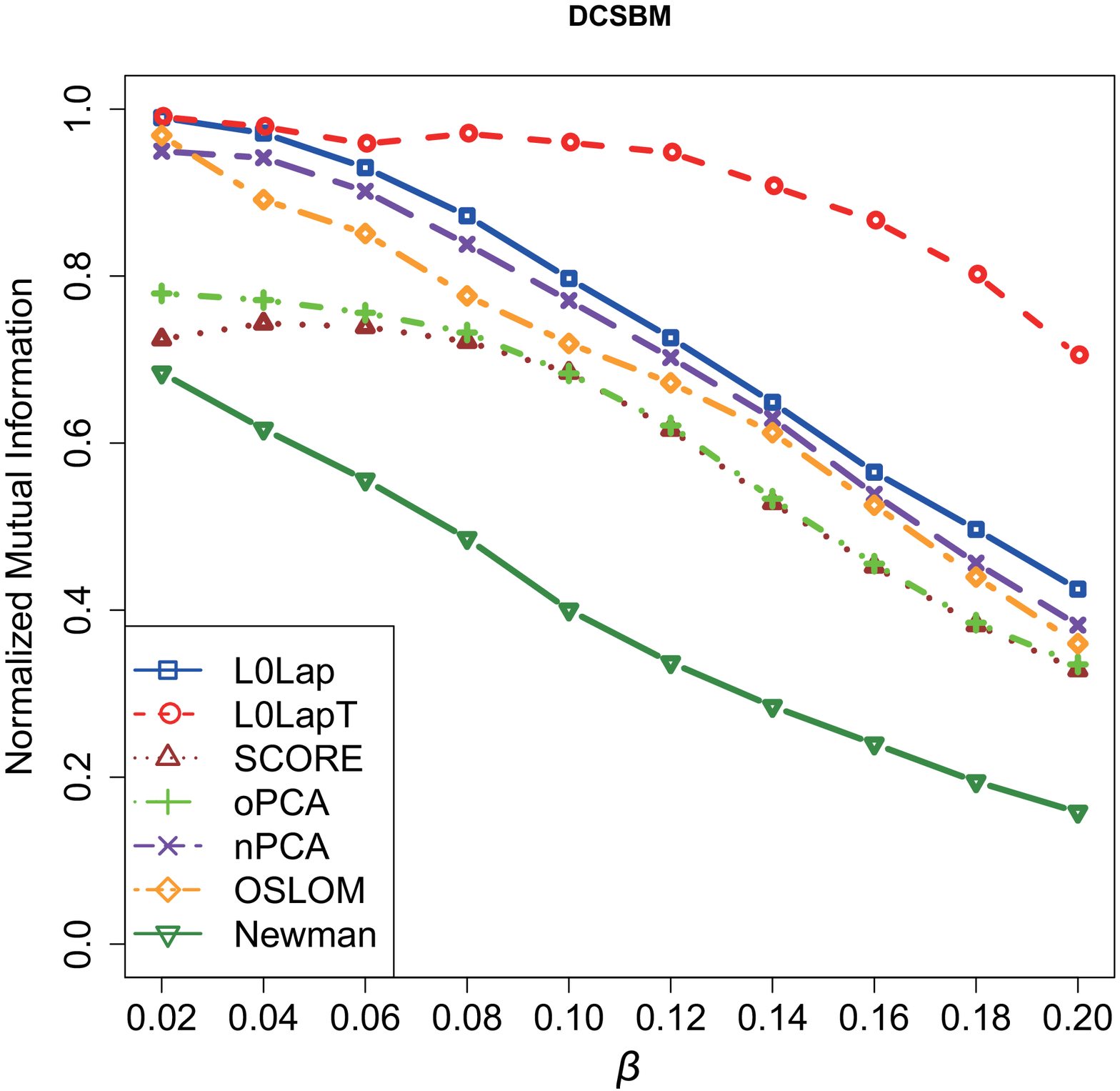}
	\end{minipage}
	\begin{minipage}[t]{0.5\linewidth}
		\centering
		\includegraphics[width=0.9\textwidth,natwidth=510,natheight=542]{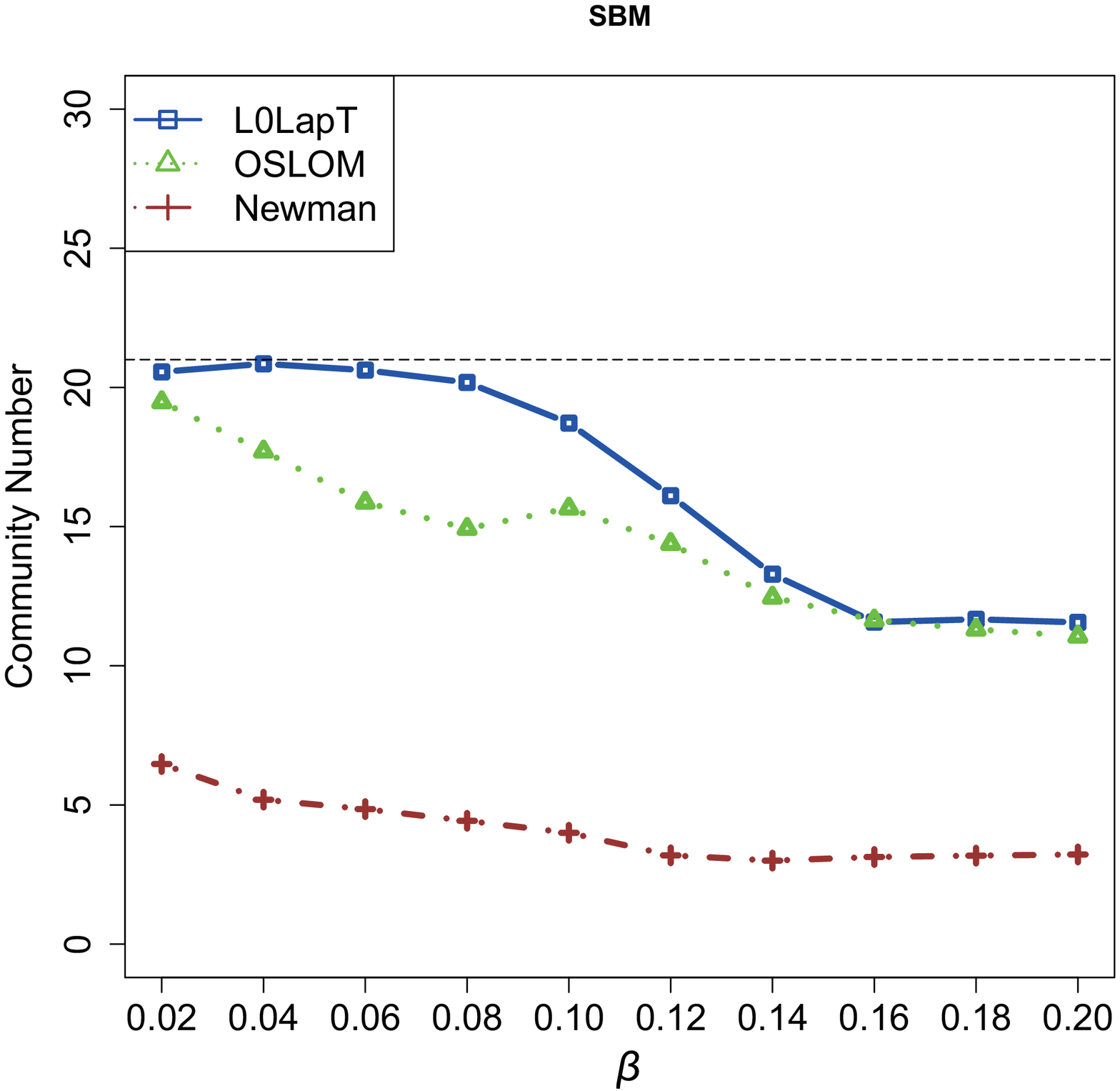}
	\end{minipage}
	\begin{minipage}[t]{0.5\linewidth}
		\centering
		\includegraphics[width=0.9\textwidth,natwidth=510,natheight=542]{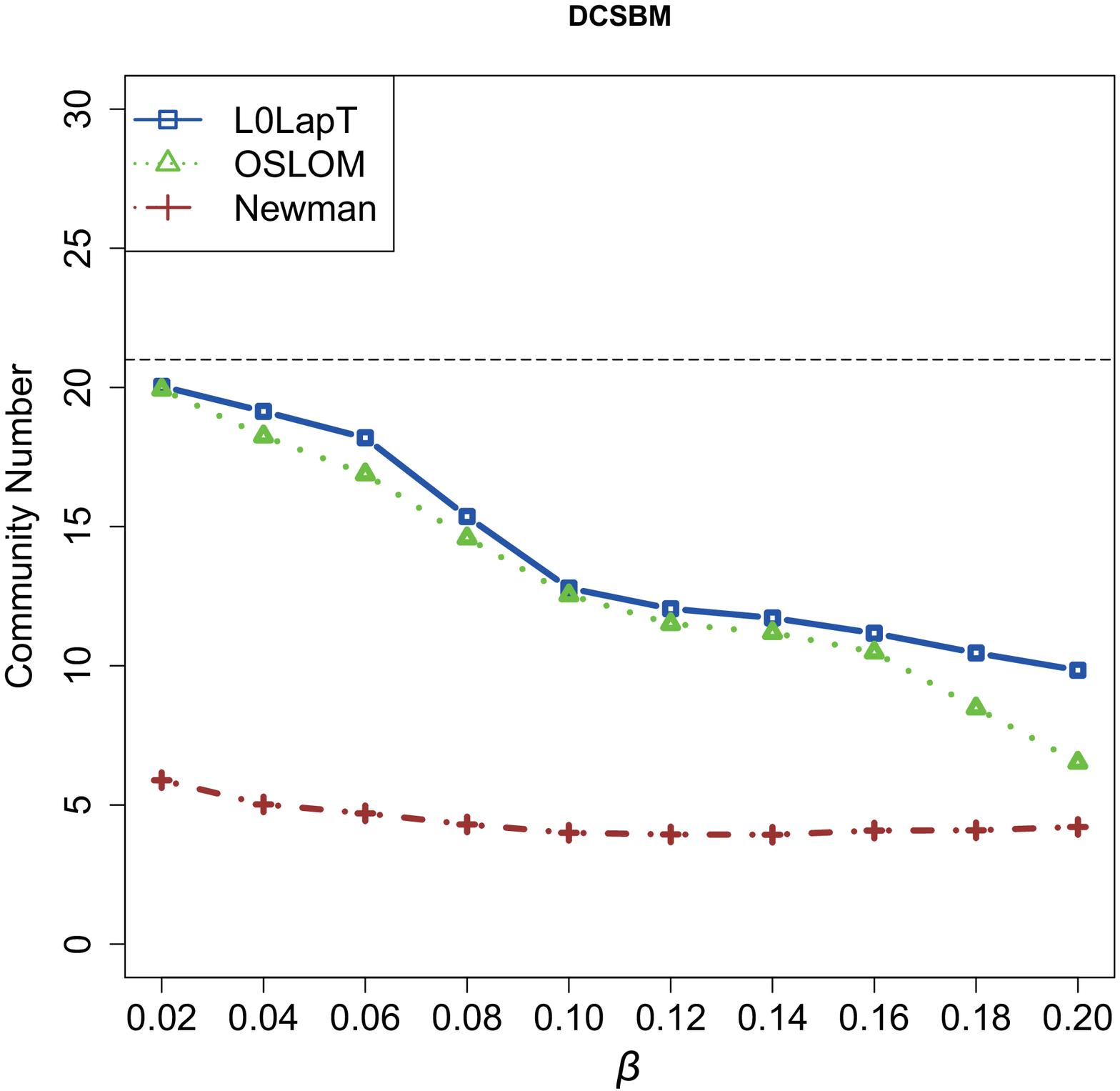}
	\end{minipage}	
	\caption{\label{cn:1}Under the SBM and DCSBM, the mean NMI (top panel) and the mean detected community number (bottom panel) over 100 simulated networks with varying out-in-ratio parameter $\beta$. The degree parameter $\Lambda$ is fixed as 50. }
\end{figure}

\subsection{Simulation under DCSBM}
\label{sec:simu:DCSBM}
In this subsection, we simulate networks under the SBM and DCSBM. All simulated networks have $n=1,000$ nodes and $K=21$ communities of different sizes. Among the 21 communities, 5 of them have 100 nodes, 6 have 50 nodes and 10 have 20 nodes. Let $\pi=\left(\pi_{1},\pi_{2},...,\pi_{21}\right)$
 be the proportion of nodes in each community (e.g. $\pi_1=\cdots=\pi_5=100/1000=0.1$, $\pi_6=\cdots=\pi_{11}=50/1000=0.05$). Conditional on the labels, the edges between nodes are generated as independent Bernoulli variables with probabilities proportional to $\theta_{i}\theta_{j}P_{ij}$. For the SBM, the parameters $\theta_{j}$ are all equal to 1. For the DCSBM, $\theta_{j}$'s are drawn independently from  $U~[0.5,1]$.

Similar to \citet{amini2013pseudo}, the connecting matrix $P$ is constructed depending on an ``out-in-ratio" parameter $\beta$ \citep{decelle2011asymptotic}. Given a $\beta$, we set the diagonal elements of a $21\times 21$ matrix $P^{(0)}$ as $\beta^{-1}$ and set all off-diagonal elements as 1. Then, given an overall expected network degree $\Lambda$, we rescale $P^{(0)}$ to give the final $P$:
\begin{equation*}
P=\frac{\Lambda}{(n-1)(\pi^{T}P^{(0)}\pi)(\mathbb{E}\Theta)^{2}}P^{(0)}.
\end{equation*}
The normalized mutual information (NMI) \citep{yao2003information} is often used to measure the concordance between the detected and true communities. Given two community labels $c$ and $e$, assume $c$ has $K_{1}$ communities and $e$ has $K_{2}$ communities. Define the confusion matrix $M_{K_{1}\times K_{2}}$ by $M_{ij}=\sum_{k}\mathbf{1}\{c_{k}=i,e_{k}=j\}$, denote its row and column sums $M_{i+}$ and $M_{+j}$. The NMI is defined by $\mbox{NMI}(c,e)=-\left(\sum_{i,j}M_{ij}\log M_{ij}\right)^{-1}\sum_{i,j}M_{ij}\log \frac{M_{ij}}{M_{i+}M_{+j}}$. For L0LapT and OSLOM, since there will be unclassified nodes, we only consider nodes that are assigned with a community label when calculating the NMI.

We first fix $\Lambda=50$ and vary the out-in-ratio parameter $\beta$ from 0.02 to 0.2. For each $\beta$, we generate 100 networks and compare the mean NMI of each algorithm over these 100 networks (Figure \ref{cn:1}). For all algorithms, the NMIs tend to decrease as the out-in-ratio parameter $\beta$ increases. We clearly see that our algorithm achieves the highest NMI compared with other methods.
As expected, after we apply the permutation test, the NMI can also be significantly improved. This is because the permutation test successfully removes small false communities. Especially, when $\beta$ is large, the test is more effective in terms of improving the NMI. For example, under the SBM, when $\beta=0.2$, the NMI of L0Lap is similar to that of nPCA and OLSOM , but after applying the permutation test, the NMI of L0Lap becomes close to 1. Furthermore, since nodes in the DCSBM are heterogeneous, as expected, all methods perform better in the SBM than in the DCSBM. We also compare the detected community numbers given by L0LapT, Newman and OSLOM (Figure \ref{cn:1} bottom panel). Our algorithm usually gives better community number estimation than the other two algorithms, although when $\beta$ is large all methods underestimate the community number. The community number chosen by Newman is much smaller than the true number. This is probably due to the known resolution limit of modularity \citep{fortunato2007resolution}.

We then fix $\beta=0.1$ and vary network degree $\Lambda$ from 2 to 100 to compare different algorithms. The mean NMI of each algorithm is shown in Figure \ref{cn:5}. Again, we see that our algorithm generally performs better than other algorithms. When $\Lambda$ is very small, since OSLOM and Newman often divide networks to many small connected subsets, they tend to have much larger NMIs than other methods, but also detect much more communities than the truth.

\begin{figure}[H]
	\begin{minipage}[t]{0.5\linewidth}
		\centering
		\includegraphics[width=0.9\textwidth,natwidth=510,natheight=542]{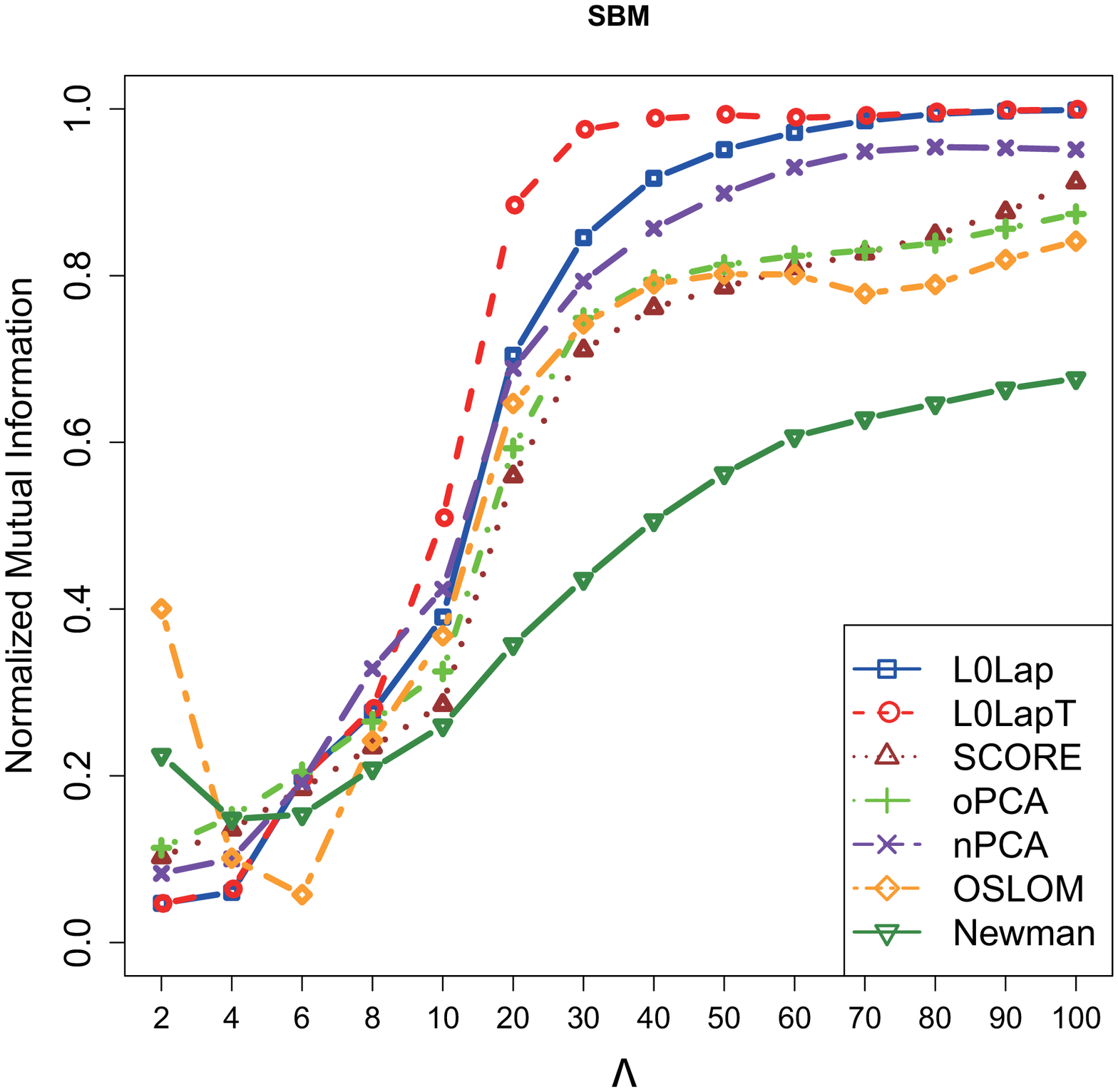}
	\end{minipage}
	\begin{minipage}[t]{0.5\linewidth}
		\centering
		\includegraphics[width=0.9\textwidth,natwidth=510,natheight=542]{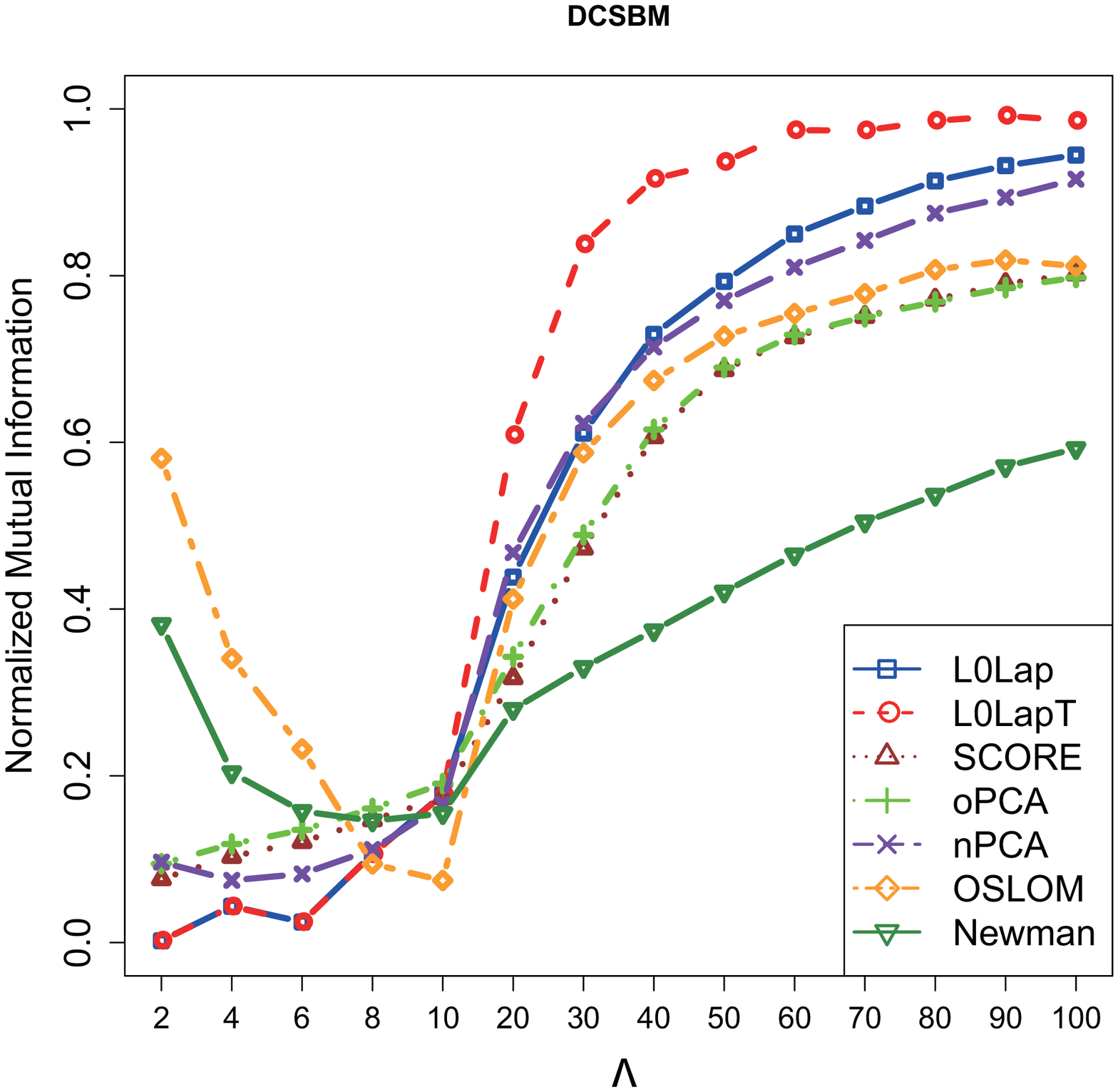}
	\end{minipage}
	
		\begin{minipage}[t]{0.5\linewidth}
		\centering
		\includegraphics[width=0.9\textwidth,natwidth=510,natheight=542]{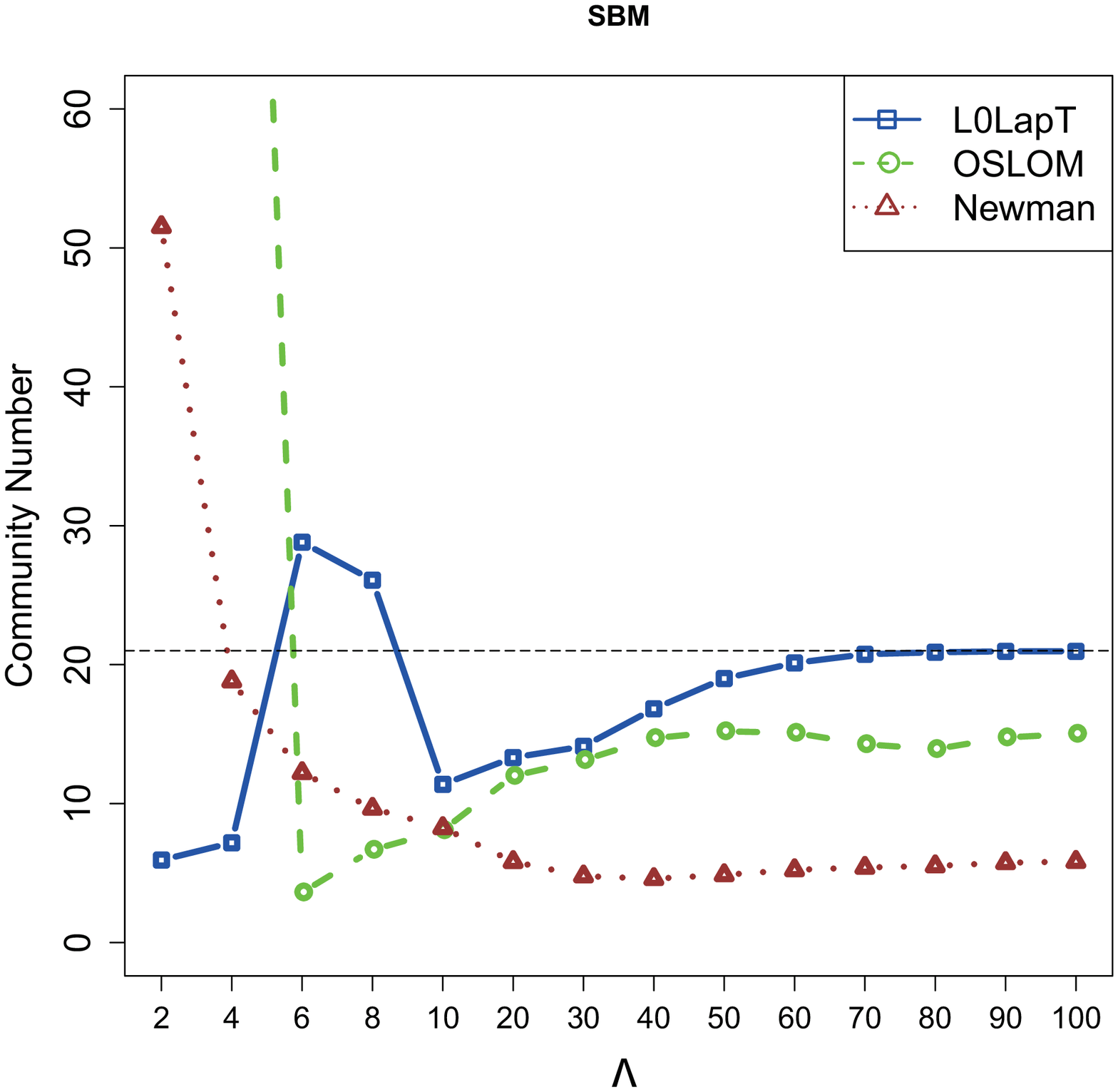}
	\end{minipage}
	\begin{minipage}[t]{0.5\linewidth}
		\centering
		\includegraphics[width=0.9\textwidth,natwidth=510,natheight=542]{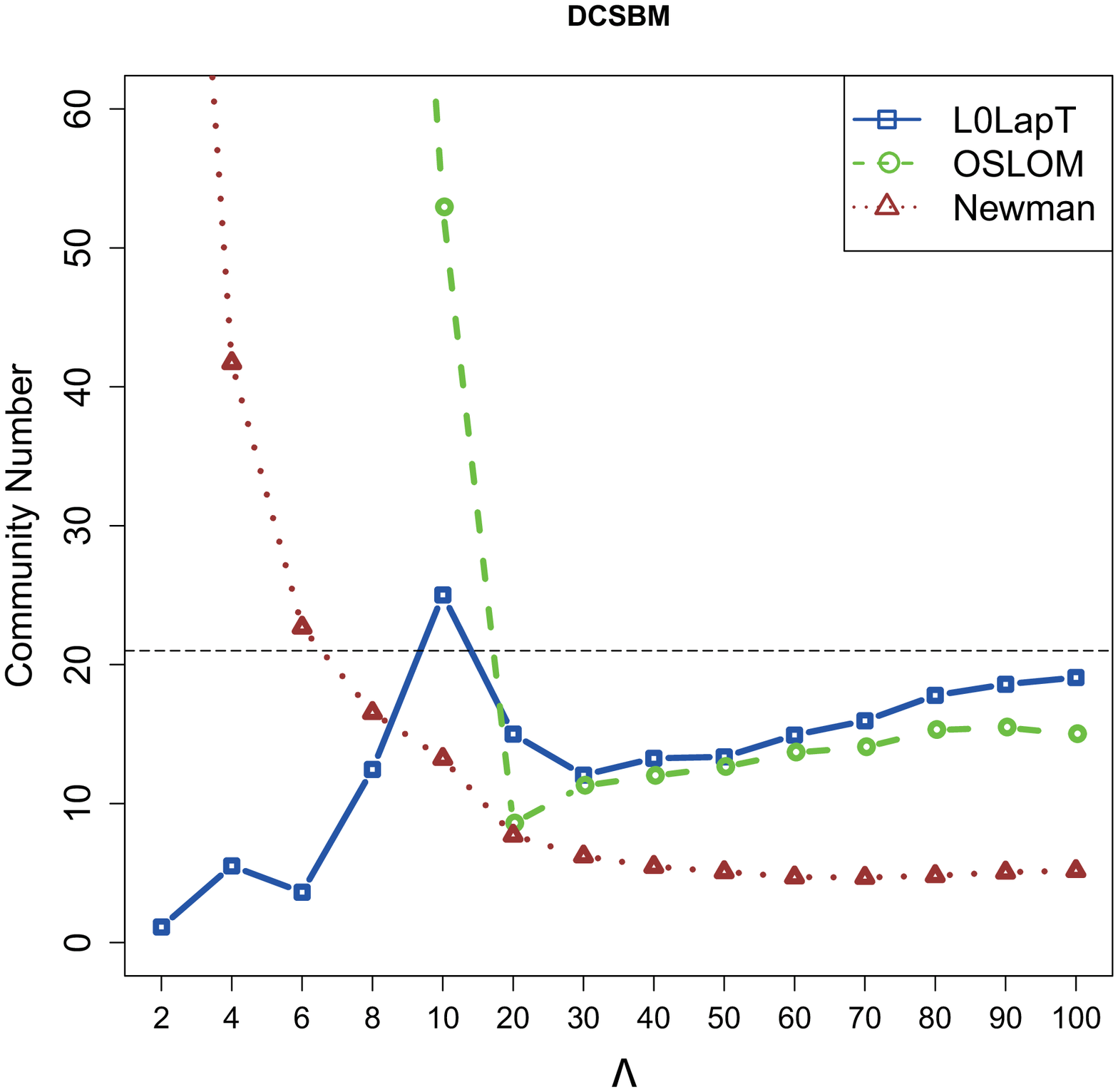}
	\end{minipage}
	\caption{\label{cn:5}Under the SBM and DCSBM, the mean NMI (top panel) and the mean detected community number (bottom panel) over 100 simulated networks with varying network degree $\Lambda$. The out-in-ratio $\beta$ is fixed as 0.1.}
\end{figure}

\subsection{Simulation under DCSBM with Outliers}

In this section, we compare the performance of each algorithm under the DCSBM with outliers. The simulated networks have $n=1,000$ nodes and $K=16$ communities of different sizes. The 16 communities are categorized into three groups according to their size, that is, 5 communities have 100 nodes, 6 have 50 nodes, and 5 have 20 nodes. The remaining 100 nodes are viewed as outliers which connect with any other nodes with the probability equal to the connecting probability between communities. The connecting probability matrix is generated similar to Section \ref{sec:simu:DCSBM}. Here, we fix $\Lambda=50$ and vary $\beta$ from 0.02 to 0.2. For nPCA, SCORE and oPCA, we set the community number as 17 in this simulation (16 communities and 1 outlier community). To compute the reasonable NMIs, the outlier nodes are viewed as in the 17th community and calculated similarly as before. Figure \ref{cn:3} shows the NMIs and the number of detected communities of these algorithms. Because there are outliers, even when $\beta$ is very small, there is still an nonignorable gap between the NMIs and its upper bound 1. However, after applying the permutation test, the NMI of L0Lap is significantly improved. Furthermore, the community number found by L0LapT is much closer to the truth compared to OSLOM and Newman.

\begin{figure}[H]
	\begin{minipage}[t]{0.5\linewidth}
		\centering
		\includegraphics[width=0.9\textwidth,natwidth=510,natheight=542]{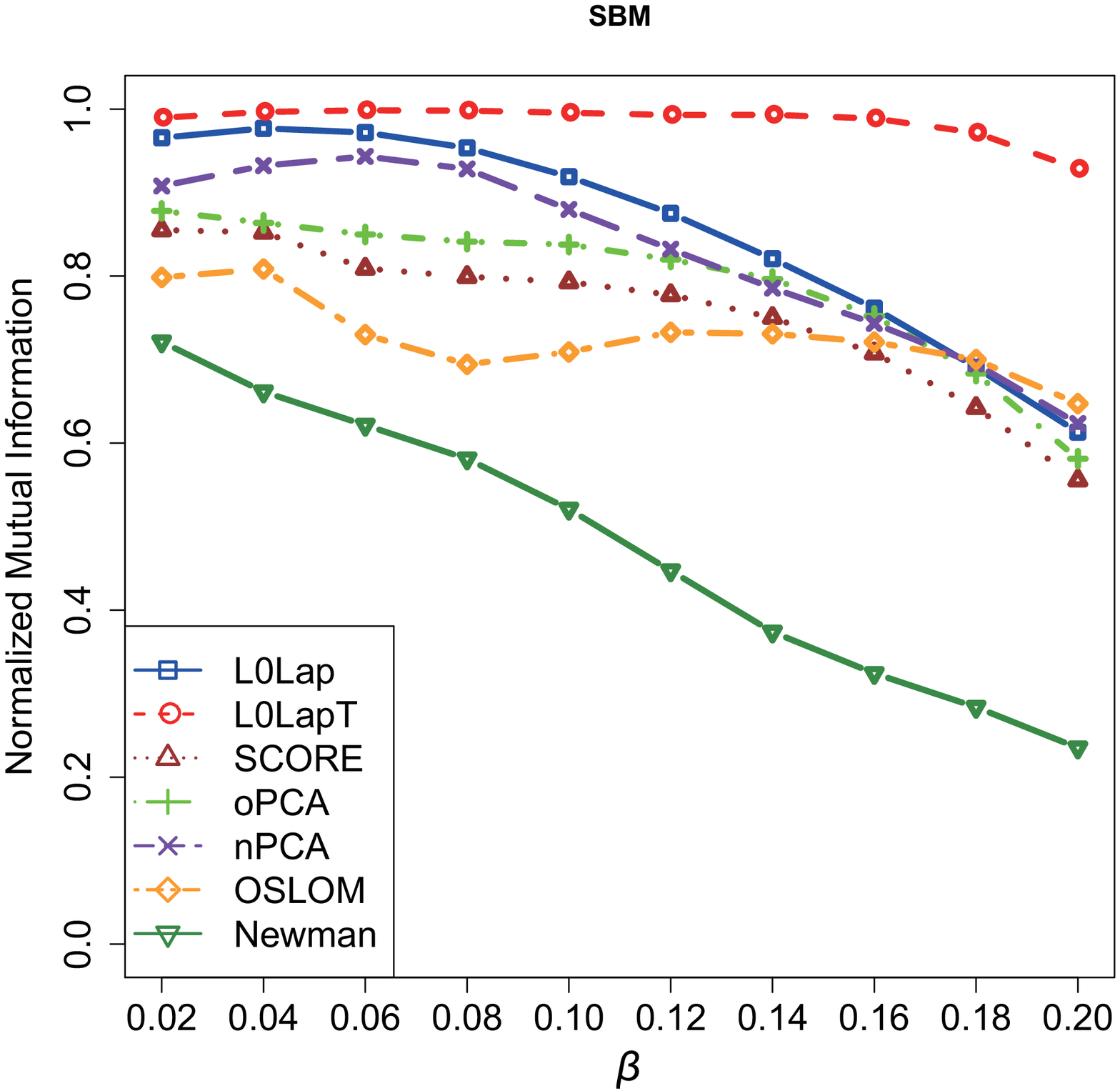}
	\end{minipage}
	\begin{minipage}[t]{0.5\linewidth}
		\centering
		\includegraphics[width=0.9\textwidth,natwidth=510,natheight=542]{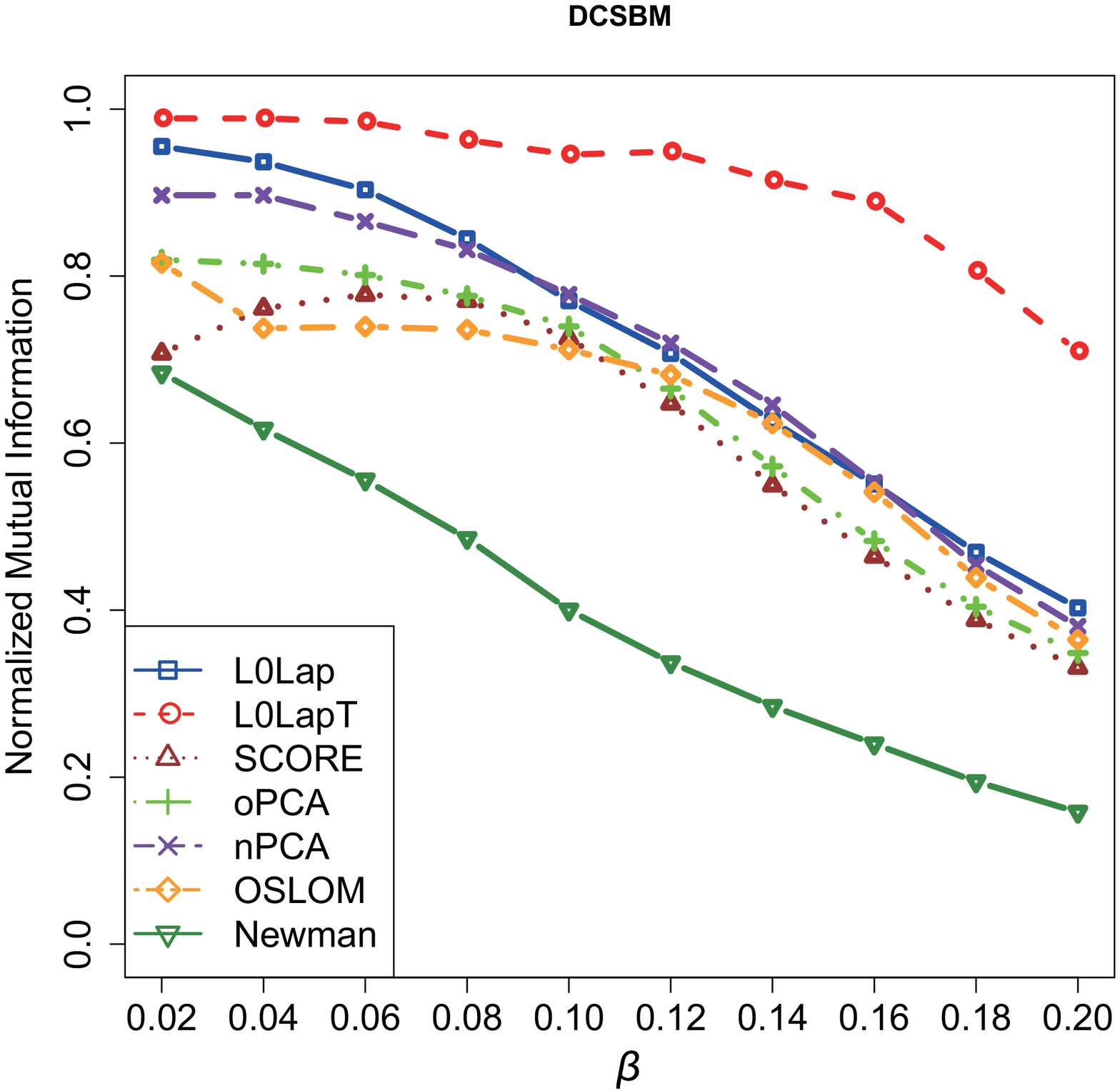}
	\end{minipage}
	\begin{minipage}[t]{0.5\linewidth}
		\centering
		\includegraphics[width=0.9\textwidth,natwidth=510,natheight=542]{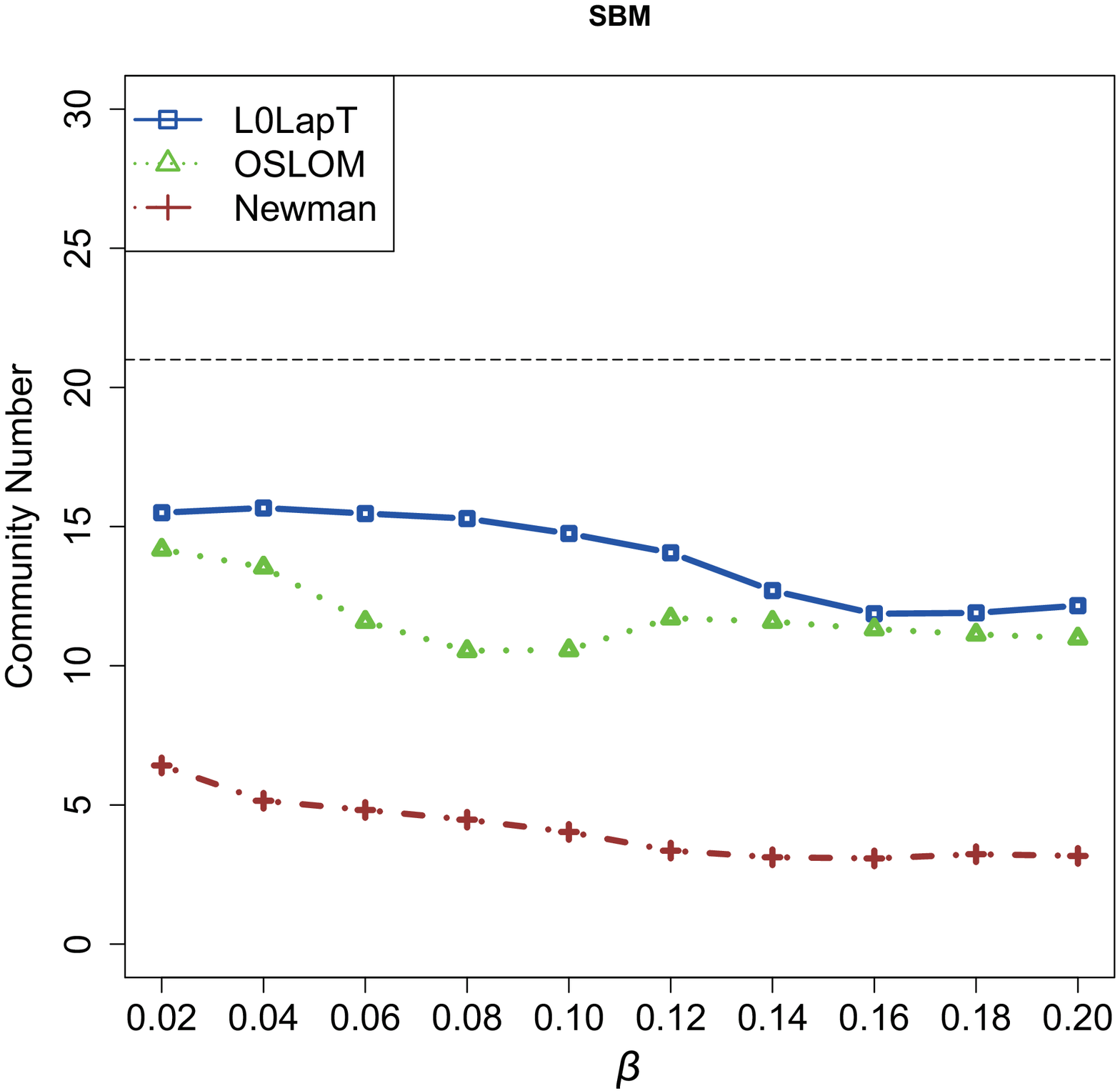}
	\end{minipage}
	\begin{minipage}[t]{0.5\linewidth}
		\centering
		\includegraphics[width=0.9\textwidth,natwidth=510,natheight=542]{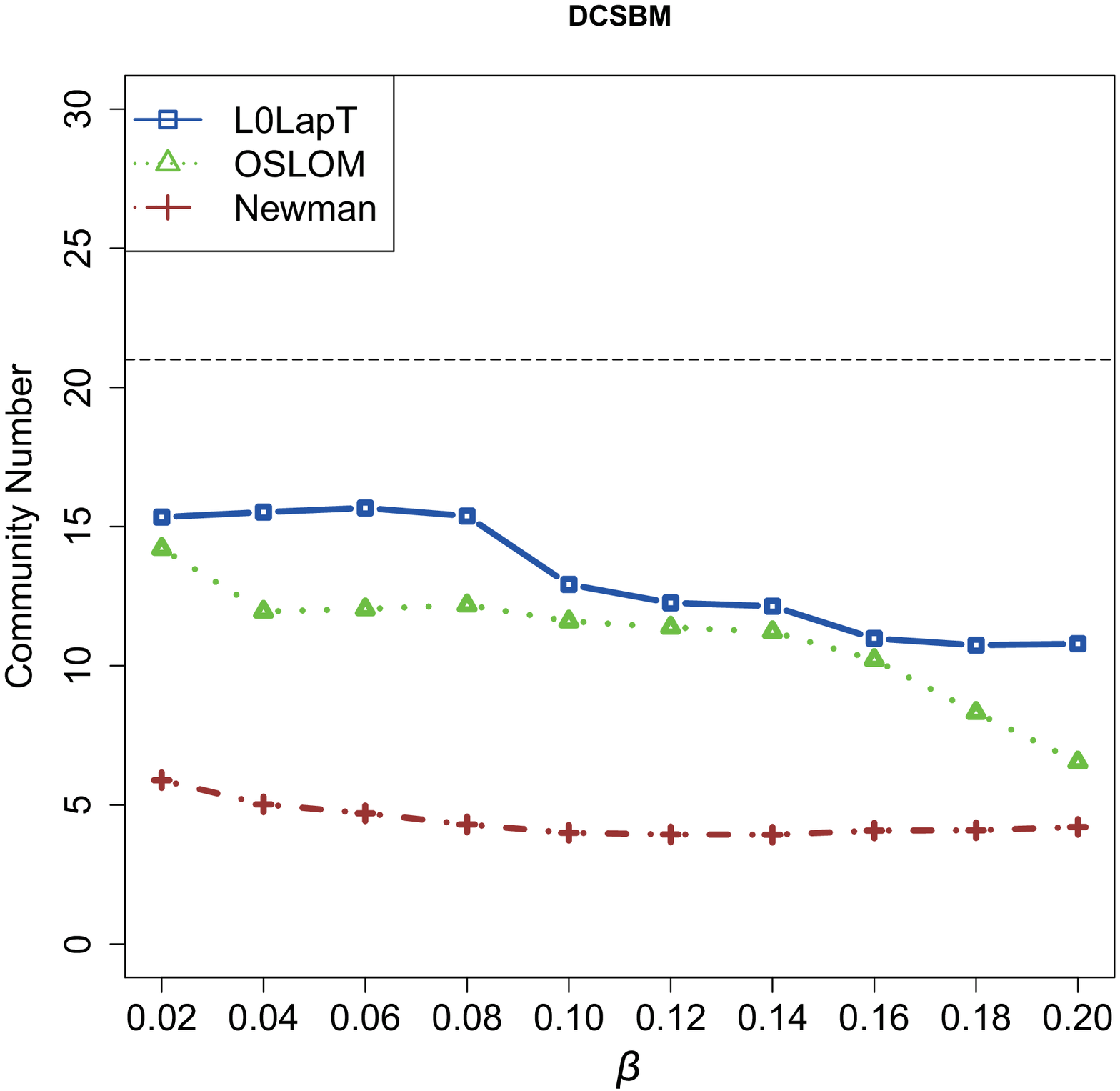}
	\end{minipage}
	\caption{\label{cn:3}Under the SBM and DCSBM with outliers, the mean NMI (top panel) and the mean detected community number (bottom panel) over 100 simulated networks with varying out-in-ratio parameter $\beta$. The degree parameter $\Lambda$ is fixed as 50.}
\end{figure}

\nop{
\subsection{Changing $\Lambda$}

In this section, our simulated network has $n=1,000$ nodes and $K=21$ communities of different sizes. The 21 communities are categorized into three groups according to their size, that is, five communities have 100 nodes, six have 50 nodes, and ten have 20 nodes. And we fix $\beta=0.1$ and change $\Lambda$ form 10 to 100. Results are further summarized by the sample mean of the 100 NMI.

It is clear to see our methods have better performance with big $\Lambda$, while about the same performance with small $\Lambda$. Newman and OSLOM have good performance when $\Lambda$ is very small, because they separate all the unicom subsets. Our methods will have similar result if we set a better initial.
}

\section{Real Data Analysis}
\label{sec:realData}
We consider two real data sets in this section, the college football network data \citep{xu2007scan} and the protein-protein network data in yeast \citep{yu2008high}.
\nop{
\subsection{Karate Club Data}
The karate club network \cite{zachary1977information} consisting of 34 nodes is often used to compare different community detection. The network is divided to two parts following a disagreement between an instructor(node 0) and an administrator(node 33), and these two groups are used as the ground truth communities in benchmark studies. The results of L0LapT, Newman and nPCA are given by Figure \ref{cn:2}. oPCA with $K=2$ almost partitions this network into the true factions. L0LapT finds three communities: the biggest community is the same with one true community except node 20, the other two communities are subsets of the other true community. In fact, if we remove node 1, we can find that the red community has no connection with the green community. It implies L0Lap can find more details than the ground truth. The remaining four grey nodes are viewed as outliers. It is reasonable since their degrees are all small. Newman also finds three communities, but it is hard to get the relationship between the red community and the green community. And the result of nPCA with $K=3$ is even worse.

\begin{figure}[H]
\centering
\includegraphics[width=0.95\textwidth]{karate.eps}
\caption{\label{cn:2}Results for the karate club network: (A) L0LapT, (B) Newman, (C) nPCA with $K=3$, (D) nPCA with $K=2$. The two shapes represent two communities by the ground truth, different colours represent different communities found by each method.}
\end{figure}
}

\subsection{College Football Data}

The college football network data is the 2006 National Collegiate Athletic Association (NCAA) Football Bowl Subdivision (FBS) schedule \citep{xu2007scan}. The data set consists of 115 schools belonging to 11 conferences in FBS, 4 independent schools and 61 lower division schools. Schools within conferences play more often against each other, so the 11 conferences are 11 communities. The four independent schools are hubs: they play against many schools in different conferences but do not belong to any conferences. The 61 lower division schools connect loosely with other nodes and are outliers of the network. We apply all methods considered in the simulation study to this data set. The algorithms L0Lap, L0LapT, OSLOM and Newman can automatically estimate the community number. For SCORE, oPCA and nPCA, we provide them with the true community number 12, including 11 communities and one outlier community. The outlier community includes both the hub nodes and the outlier nodes. Table \ref{table:1} shows the NMI and the detected community number (CN) of each algorithm. This clearly show that L0LapT have the largest NMI compared with other methods. oPCA also works well: Its NMI is 0.925 and ranks the second best among all algorithms. In terms of outlier identification, although OSLOM is designed to be able to identify outliers, it fails to report any outlier for this data. In comparison, L0LapT identifies 80 nodes as outliers and 62 of them are true outliers. OSLOM gives the most accurate estimate of the community number.

\begin{table}[H]
\label{simulation}
\caption{\label{table:1}The performance of all methods on the college football network data. CN is the detected or the provided community number. We set the community number as 12 for SCORE, oPCA and nPCA.}

\begin{center}
\begin{tabular}{c|ccccccc}
\hline

\multicolumn{ 1}{c}{} &      \multicolumn{ 1}{c}{L0Lap}  &      \multicolumn{ 1}{c}{L0LapT} &               \multicolumn{ 1}{c}{SCORE} &      \multicolumn{ 1}{c}{oPCA}  &      \multicolumn{ 1}{c}{nPCA} &      \multicolumn{ 1}{c}{OSLOM}  &      \multicolumn{ 1}{c}{Newman} \\
\hline

\multicolumn{ 1}{c}{NMI} &       0.856  &     0.985  &  0.674  &  0.925  &  0.640  &  0.681  &  0.550 \\ \hline

\multicolumn{ 1}{c}{CN} &       22  &     10  &     12 &  12  &  12  &  11 &  6  \\

\hline
\end{tabular}
\end{center}
\end{table}

To look into more details of the detected communities of each algorithm, we examine the pairwise overlaps between detected communities with true communities. Specifically, given a detected community $C^{D}_i$ and a true community $C^{T}_j$, we calculate an overlapping score between these two communities by $o_{ij}=|C^{D}_i\bigcap C^{T}_j|\big/|C^{D}_i\bigcup C^{T}_j|$. Thus, we get a matrix $O=(o_{ij})_{CN\times 12}$ for each algorithm. Figure \ref{cn:6} shows heat maps of these matrices for L0LapT, oPCA, Newman and OSLOM. The true community 12 in Figure \ref{cn:6} is the outlier community. All communities identified by L0LapT are highly similar to or exactly the same as the true communities. This demonstrates that L0LapT can give high quality communities. However, L0LapT fails to detect the community 11 and nodes in this community are filtered as outliers. In comparison, although OSLOM gives the best estimate of the community number, the quality of its detected communities are not as high as those given by L0LapT. Most of the ``diagonal" overlapping scores for OSLOM are less than 0.71 and the largest overlapping score is only 0.86, showing that many detected communities contain substantial amount of nodes not belonging to these communites. oPCA performs well for most communities, but members in community 2 and community 11 are mixed up. Newman performs poorly in this data. Most of its detected communities are far away from true communities.
\begin{figure}[H]
\centering
\includegraphics[width=1\textwidth]{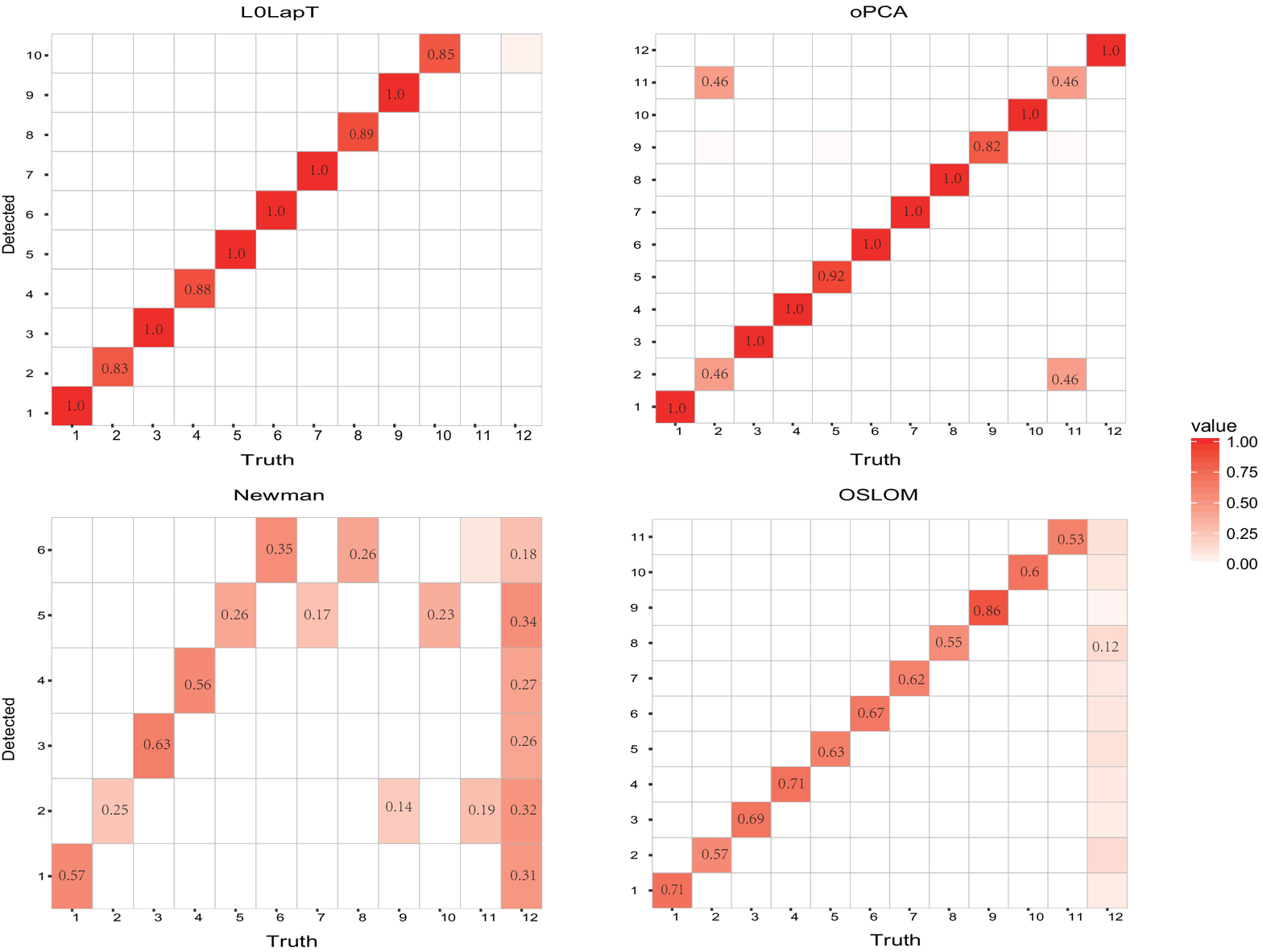}
\caption{\label{cn:6} Heatmap of the overlapping scores $o_{ij}$ between the detected communities with the true communities for the college football network data. The figure only shows results for L0LapT, oPCA, Newman and OSLOM. Similar plots for other algorithms are shown in the Supplementary Figure. The numbers in the figure are the overlapping scores $o_{ij}$ with $o_{ij}>0.1$.}
\end{figure}

\subsection{Protein-protein interaction data in yeast}

In this section, we consider a protein-protein interaction network data in yeast \citep{yu2008high}. After removing isolated nodes, we get a network with 1,540 nodes and 7,123 edges. Different proteins often interact with each other to achieve one biological function. The communities of the PPI network should then represent different cellular functions.  We apply all methods considered in the simulation study to this network to find communities in this PPI network. L0LapT finds 22 communities with their sizes ranging from 8 to 138. OSLOM finds 114 communities ranging from 3 to 103. Newman finds 202 communities ranging from 2 to 193. For SCORE, oPCA and nPCA, since the number of communities is unknown, we set the community number as 100, which roughly is the average number of communities detected by L0LapT, OSLOM and Newman. Finally, the community sizes given by SCORE, oPCA and nPCA ranges from 1 to 1178, from 1 to 738 and 1 to 1069, respectively. We further filter out communities with $\leq 5$ nodes, since these are unlikely to be true communities.

Since we do not know the true community structure, to evaluate the quality of the partition of this yeast network, we instead use gene oncology (GO) enrichment analysis to compare different algorithms. Since communities of the PPI network correspond to different cellular functions, the detected communities should be enriched with known GO terms. We download yeast gene GO annotation database from \url{http://www.yeastgenome.org/}. For enrichment analysis, we focus only on GO terms with at least 10 annotated genes. For each community, we calculate a list of p-values with every GO term by Fisher's exact test. If the detected communities are biological meaningful, the communities should be highly significant with a number of GO terms. After $\log_{10}$ transformation of these p-values, define ${\rm ratio}_{t}={|-\log_{10}\mbox{p-value}>t|}\big/{|-\log_{10}\mbox{p-value}>0|}$ for a threshold $t$. This ratio could be viewed as an indicator of biological relatedness of the detected communities. At the same cutoff $t$, larger ratio value should correspond to more biologically meaningful communities. The ratio curves of these methods are shown in Figure \ref{GO enrichment1}, left panel. We see that the curve of L0LapT is largely above other curves. However, when $t$ is large, it is hard to see the difference. Therefore, we further consider only p-values less than 0.1 and define ${\rm ratio}^r_{t}={|\{-\log_{10}\mbox{p-value}>t\}|}\big/{|\{-\log_{10}\mbox{p-value}>1\}|}$ for any threshold $t\geq 1$. The new ratio curves  are shown in Figure \ref{GO enrichment1}, right panel. We can now clearly see that the curve of L0LapT curves is always above other methods.
\begin{figure}[H]
	\begin{minipage}[t]{0.5\linewidth}
		\centering
		\includegraphics[width=1.1\textwidth,natwidth=610,natheight=642]{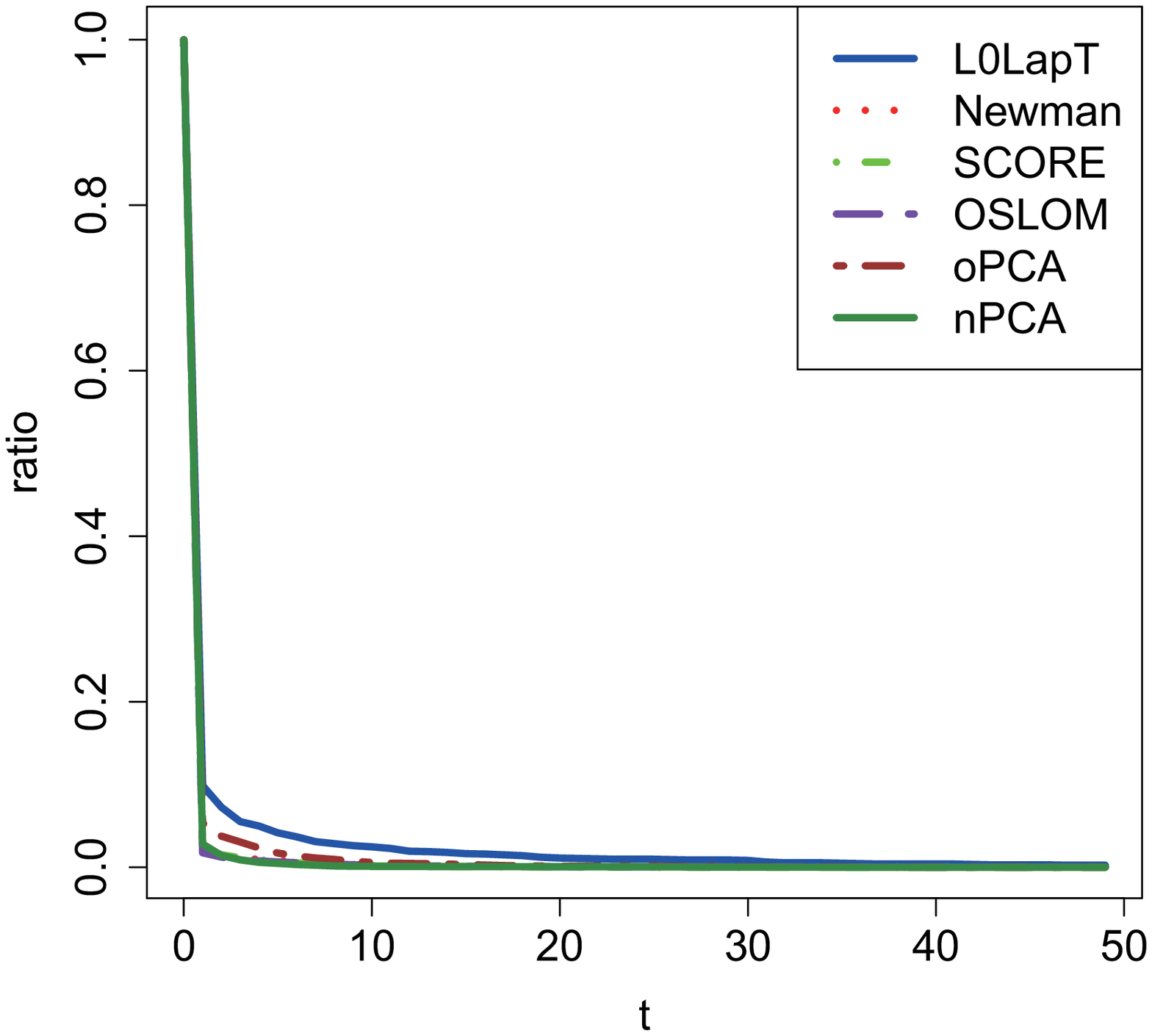}
	\end{minipage}
	\begin{minipage}[t]{0.5\linewidth}
		\centering
		\includegraphics[width=1.1\textwidth,natwidth=610,natheight=642]{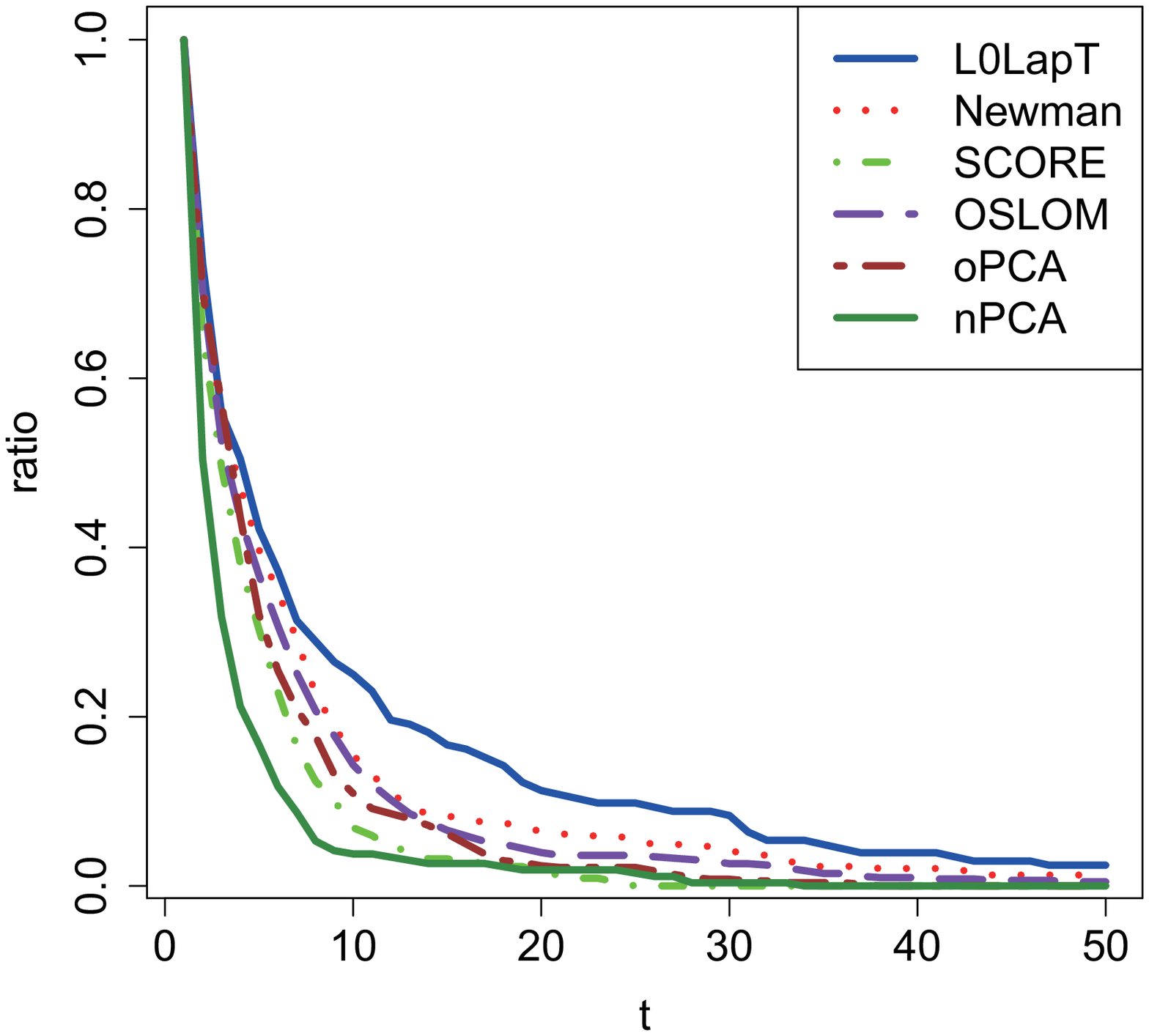}
	\end{minipage}
	\caption{\label{GO enrichment1}GO enrichment analysis.}
\end{figure}

\section{Conclusion and Discussion}
In this paper, we propose a $L_0$-penalized Laplacian for community detection. This method does not require information about the community number and it can detect communities in networks with outliers. We prove a consistency result for DCSBM with or without outliers. Simulation studies show that the proposed method generally performs better than other available algorithms. One problem we found is that although the proposed method generally gives more accurate estimation of the community number, when networks contain more noise or when the network is too sparse, the proposed algorithm still cannot give a very accurate estimate of community number. In addition, the statistical test used in this paper is based on permutation. Although simulation shows that this permutation works well in general in terms filtering false communities, we were not able to develop theoretical guarantees for this test. This seems a quite difficult question. As far as we know, there is currently no theory developed for permutation tests in networks.

\section{Appendix}
\label{sec:proof}

In this section, we give proofs of our theoretical results. Before proving the main theorem, we first give some lemmas.

\begin{lemma} \label{Lemma_ConExp}
	Under the assumptions of DCSBM, we have
	$$\mathbb{E}\left(W(S)|\mathbf{c}\right)=\sum_{k=1}^{K}nr_{k}^{d}(S)\left(\sum_{l=1}^{K}nr_{l}^{d}(S)p_{kl}\right)~\mbox{and}~~\mathbb{E}(V(S)|\mathbf{c})=\sum_{k=1}^{K}nr_{k}^{d}(S)\left(\sum_{l=1}^{K}n\hat{\pi}_{l}^{d}p_{kl}\right)$$
\end{lemma}

\begin{proof}
	Under the assumptions of DCSBM, we have
	\begin{equation*}
	\mathbb{E}(A_{ij}|c_{i}=k,c_{j}=l)=\mathbb{E}(\theta_{i}|c_{i}=k)\mathbb{E}(\theta_{j}|c_{j}=l)p_{kl}=
	\frac{\pi_{k}^{d}}{\pi_{k}}\frac{\pi_{l}^{d}}{\pi_{l}}p_{kl}.
	\end{equation*}
	
	So we have
	\begin{align*}
	\mathbb{E}(W(S)|\mathbf{c}) &= \sum_{k=1}^{K}\sum_{l=1}^{K}\sum_{i \in S_{k}, j\in S_{l}}\mathbb{E}(A_{ij}|\mathbf{c})\\
	&= \sum_{k=1}^{K}\sum_{l=1}^{K}\sum_{i \in S_{k}, j\in S_{l}}\frac{\pi_{k}^{d}}{\pi_{k}}\frac{\pi_{l}^{d}}{\pi_{l}}p_{kl}
	= \sum_{k=1}^{K}nr_{k}^{d}(S)(\sum_{l=1}^{K}nr_{l}^{d}(S)p_{kl}),
	\end{align*}
and	
	\begin{align*}
	\mathbb{E}(V(S)|\mathbf{c}) &= \sum_{k=1}^{K}\sum_{l=1}^{K}\sum_{i \in S_{k},j\in G_{l}}\mathbb{E}(A_{ij}|\mathbf{c})\\
	&= \sum_{k=1}^{K}\sum_{l=1}^{K}\sum_{i \in S_{k},j\in G_{l}}\frac{\pi_{k}^{d}}{\pi_{k}}\frac{\pi_{l}^{d}}{\pi_{l}}p_{kl}
	= \sum_{k=1}^{K}nr_{k}^{d}(S)(\sum_{l=1}^{K}n\hat{\pi}_{l}^{d}p_{kl}).
	\end{align*}
	
\end{proof}

We need Chernoff's inequality \citep{furedi1981eigenvalues} and Hoeffding's inequality \citep{hoeffding1963probability} to prove Theorem \ref{DCSBM}.
\begin{lemma}
	(Chernoff's inequality) Let $X_{1},...,X_{n}$ be independent random variables with $$\mathbb{P}(X_{i}=1)=p_{i}, \ \mathbb{P}(X_{i}=0)=1-p_{i}.$$ Then the sum $X=\sum_{i=1}^{n}X_{i}$ has expectation $\mathbb{E}(X)=\sum_{i=1}^{n}p_{i}$, and we have $$\mathbb{P}\left(X<\mathbb{E}\left(X\right)-\lambda\right)<\exp\big\{-{2^{-1}\lambda^{2}}/{\mathbb{E}(X)}\big\},$$ $$\mathbb{P}\left(X>\mathbb{E}\left(X\right)+\lambda\right)<\exp\big\{-{2^{-1}\lambda^{2}}/{(\mathbb{E}(X)+\lambda/3)}\big\}.$$
\end{lemma}

\begin{lemma}
	(Hoeffding's inequality) Let $X_{1},...,X_{n}$ be independent random variables and $X_{i}$'s are strictly bounded by the intervals $[a_{i},b_{i}]$. We define the empirical mean of these variables by
	$\bar{X}=n^{-1}\sum_{i=1}^{n}X_{i},$
	then we have
	$$\mathbb{P}\left(\left|\bar{X}-\mathbb{E}(\bar{X})\right|>t\right)\leq 2\exp\bigg\{-\frac{2n^{2}t^{2}}{\sum_{i=1}^{n}(b_{i}-a_{i})^{2}}\bigg\}.$$
\end{lemma}

\begin{lemma}
	\label{cd}
	Define $\hat{\psi}(S)={\mathbb{E}(W(S)|\mathbf{c})}/{\mathbb{E}(V(S)|\mathbf{c})}-\eta|S|.$ Under the assumptions of DCSBM, we have
	$$\max_{S\in \Gamma_{\delta}}\left|\psi(S)-\hat{\psi}(S)\right|\lesssim {n^{\delta-\alpha}}$$
	with probability at least $1- 2^{n+2}/n^n$ when $n$ is sufficiently large.
\end{lemma}
\begin{proof}
	By Lemma \ref{Lemma_ConExp} and the condition $p^{-}\gtrsim {\log n}/{n^{1-2\alpha}}$, we have
	\begin{align}
	\label{WS_ExpectationIneq}
	\mathbb{E}(W(S)|\mathbf{c}) &=\sum_{k=1}^{K}nr_{k}^{d}(S)(\sum_{l=1}^{K}nr_{l}^{d}(S)p_{kl}) \nonumber \\
	&\geq p^{-} \sum_{k=1}^{K}n^2(r_{k}^{d}(S))^{2} \gtrsim n^{1+2\alpha}\log n \sum_{k=1}^{K}(r_{k}^{d}(S))^{2}.
	\end{align}
	Since $\sum_{k=1}^Kr_k^{d}(S) = \sum_{k=1}^Kr_{k}(S){\pi_{k}^{d}}/{\pi_{k}}$ and $0<h_{1}\leq {\pi_{k}^{d}}/{\pi_{k}}\leq h_{M}$, we have $\sum_{k=1}^{K}(r_{k}^{d}(S))^{2}\geq {(\sum_{k=1}^Kr_k^{d}(S))^{2}}/{K}\geq h_{1}^{2}{|S|^{2}}/(Kn^{2})$ by the Cauchy-Schwarz inequality. Then we have $\mathbb{E}(W(S)|\mathbf{c}) \gtrsim n^{1+2(\alpha-\delta)}\log n$ if $S\in \Gamma_{\delta}$. Let $\lambda=2\sqrt{n\log n\mathbb{E}(W(S)|\mathbf{c}})$ and by Chernoff's inequality, we have
	\begin{equation*}
	\mathbb{P}\left(W(S)-\mathbb{E}(W(S)|\mathbf{c})<-\lambda\right)<n^{-n}.
	\end{equation*}
	Since $\mathbb{E}(W(S)|\mathbf{c})\gtrsim n^{1+2(\alpha-\delta)}\log n$, we have $\lambda/3<\mathbb{E}(W(S)|\mathbf{c})$ with sufficiently large $n$ and thus
	\begin{equation*}
	\mathbb{P}\left(W(S)-\mathbb{E}(W(S)|\mathbf{c})>\lambda\right)<n^{-n}.
	\end{equation*}
	So we have
	\begin{equation*}
	\mathbb{P}\left(\left|W(S)-\mathbb{E}(W(S)|\mathbf{c})\right|>\lambda\right)<2n^{-n}.
	\end{equation*}
For $V(S)$, we have $$\mathbb{E}(V(S)|\mathbf{c}) \geq\mathbb{E}(W(S)|\mathbf{c}) \gtrsim n^{1+2(\alpha-\delta)}\log n$$
Similarly, let $\tilde{\lambda}=2\sqrt{n\log n\mathbb{E}(V(S)|\mathbf{c})}$, and we have
	\begin{equation*}
	\mathbb{P}\left(\left|V(S)-\mathbb{E}(V(S)|\mathbf{c})\right|>\tilde{\lambda}\right)<2n^{-n}.
	\end{equation*}
In addition, we have
	$$\frac{\lambda}{\mathbb{E}(W(S)|\mathbf{c})}\lesssim \frac{1}{n^{\alpha-\delta}}~\mbox{and}~~\frac{\tilde{\lambda}}{\mathbb{E}(V(S)|\mathbf{c})}\lesssim \frac{1}{n^{\alpha-\delta}}.$$
Thus, with probability at least $1- 4/n^n$, we have
	\begin{align*}
	&\left|\frac{W(S)}{V(S)}-\frac{\mathbb{E}(W(S)|\mathbf{c})}{\mathbb{E}(V(S)|\mathbf{c})}\right|\\ &\leq\max\left\{\left|\frac{\mathbb{E}(W(S)|\mathbf{c})-\lambda}{\mathbb{E}(V(S)|\mathbf{c})+\tilde{\lambda}}-\frac{\mathbb{E}(W(S)|\mathbf{c})}{\mathbb{E}(V(S)|\mathbf{c})}\right|,
	\left|\frac{\mathbb{E}(W(S)|\mathbf{c})+\lambda}{\mathbb{E}(V(S)|\mathbf{c})-\tilde{\lambda}}-\frac{\mathbb{E}(W(S)|\mathbf{c})}{\mathbb{E}(V(S)|\mathbf{c})}\right|\right\}\\
	&=\max\left\{\left|\frac{\mathbb{E}(W(S)|\mathbf{c})\tilde{\lambda}+\mathbb{E}(V(S)|\mathbf{c})\lambda}{\mathbb{E}(V(S)|\mathbf{c})(\mathbb{E}(V(S)|\mathbf{c})+\tilde{\lambda})}\right|
	,\left|\frac{\mathbb{E}(W(S)|\mathbf{c})\tilde{\lambda}+\mathbb{E}(V(S)|\mathbf{c})\lambda}{\mathbb{E}(V(S)|\mathbf{c})(\mathbb{E}(V(S)|\mathbf{c})-\tilde{\lambda})}\right|\right\}\\
	&\leq \left|\frac{\lambda}{\mathbb{E}(V(S)|\mathbf{c})-\tilde{\lambda}}\right|+\left|\frac{\tilde{\lambda}}{\mathbb{E}(V(S)|\mathbf{c})-\tilde{\lambda}}\right| \lesssim \frac{1}{n^{\alpha-\delta}}.
	\end{align*}
	Therefore, with probability at least $1- 2^{n+2}/n^n$ we have
	$$\max_{S\in \Gamma_{\delta}}\left|\psi(S)-\hat{\psi}(S)\right|\lesssim \frac{1}{n^{\alpha-\delta}},$$
	when $n$ is sufficiently large.
\end{proof}

\begin{lemma}
	\label{rho}
	For DCSBM, with probability at least $1-2Kn^{-2}$, we have
	$$|\hat{\rho}_{k}^{d}-\rho_{k}^{d}|\lesssim \frac{1}{n^{\alpha-\delta}},$$
	for all $1\leq k\leq K$.
\end{lemma}

\begin{proof}
	By definition, we have
	$$\hat{\pi}_{k} = \frac{1}{n}\sum_{i=1}^{n}\mathbf{I}\{c_{i}=k\}.$$
	Since $\mathbb{E}(\hat{\pi}_{k})=\pi_{k}$ and $\mathbf{I}\{c_{i}=k\}$ are strictly bounded by the intervals $[0,1]$, we have
	$$\mathbb{P}\left(\left|\hat{\pi}_{k}-\pi_{k}\right|>\sqrt{\frac{\log n}{n}}\right)\leq \frac{2}{n^{2}},$$
	by Hoeffding's inequality.
	
Since $\hat{\pi}_{k}^{d} = {\pi_{k}^{d}}\hat{\pi}_{k}/{\pi_{k}}$ and ${\pi_{k}^{d}}/{\pi_{k}}\leq h_{M}$, we have
	$$\mathbb{P}\left(\left|\hat{\pi}_{k}^{d}-\pi_{k}^{d}\right|>h_{M}\sqrt{\frac{\log n}{n}}\right)\leq \frac{2}{n^{2}},$$
	by the inequality above. Therefore, with probability at least $1-2Kn^{-2}$ we have
	\begin{align*}
	|\hat{\rho}_{k}^{d}-\rho_{k}^{d}|&=\left|\frac{p_{kk}}{\sum_{l=1}^{K}\pi_{l}^{d}p_{kl}}-\frac{p_{kk}}{\sum_{l=1}^{K}\hat{\pi}_{l}^{d}p_{kl}}\right|
	=p_{kk}\left|\frac{\sum_{l=1}^{K}(\hat{\pi}_{l}^{d}-\pi_{l}^{d})p_{kl}}{(\sum_{l=1}^{K}\pi_{l}^{d}p_{kl})(\sum_{l=1}^{K}\hat{\pi}_{l}^{d}p_{kl})}\right|\\
	&\lesssim \sqrt{\frac{\log n}{n}}\frac{1}{(\pi^{-})^{2}}
	\lesssim \frac{\sqrt{\log n}}{n^{1/2-\delta}}
	\lesssim \frac{1}{n^{\alpha-\delta}}.
	\end{align*}
\end{proof}

\begin{lemma}\label{Ceta}
	Assume real numbers $0<x_{k},y_{k},z_{kl}\leq 1$ satisfy $0<C_{1}\leq{x_{k}}/{y_{k}}\leq C_{2}$ for all $1\leq k\neq l \leq K$. Define
	\begin{equation}
	f(t_{1},...,t_{K})=\frac{\sum_{k=1}^{K}t_{k}(t_{k}x_{k}+\sum_{l\neq k}t_{l}z_{kl})}{\sum_{k=1}^{K}t_{k}y_{k}},
	\end{equation}
	where $t_{k}\geq 0$ and $\sum_{k=1}^{K}t_{k}=1$.
	If ${x_{1}}/{y_{1}}> \max_{2\leq k\leq K}{x_{k}}/{y_{k}}$ and $\min_{1\leq k\leq K} x_{k}>max_{k\neq l}z_{kl}$, we have
	\begin{enumerate}[(1)]
		\item $f(t_{1},...,t_{K})\leq f(1,0,...,0)=\frac{x_{1}}{y_{1}},$
		\item For any $0<t<1$,$$f(1,0,...,0)-\max_{t_{1}\leq 1-t}f(t_{1},...,t_{K})\geq \frac{1}{2}\left(\frac{x_{1}}{y_{1}}-\max_{2\leq k\leq K}\frac{x_{k}}{y_{k}}\right)t.$$
	\end{enumerate}
\end{lemma}
\begin{proof}
	\begin{enumerate}[(1)]
		\item Since ${x_{1}}/{y_{1}}>\max_{2\leq k\leq K}{x_{k}}/{y_{k}}$ and $\min_{1\leq k\leq K} x_{k}>max_{k\neq l}z_{kl}$,
		\begin{align*}
		f(t_{1},...,t_{K})&\leq \frac{\sum_{k=1}^{K}t_{k}x_{k}}{\sum_{k=1}^{K}t_{k}y_{k}}
		\leq \frac{x_{1}}{y_{1}}=f(1,0,...,0).
		\end{align*}
		\item Since $f(t_{1},...,t_{K})$ is continuous and $\{(t_{1},...,t_{K})|t_{1}\leq 1-t\}$ is a close set,  $f(t_{1},...,t_{K})$ can achieve its upper bound on $\{(t_{1},...,t_{K})|t_{1}\leq 1-t\}$. Suppose that $f(t_{1},...,t_{K})$ achieves its upper bound at $(t_{1}^{*},...,t_{K}^{*})$ and define $\bar{x},\bar{y}$ such that $(1-t_{1}^{*})\bar{x}=\sum_{k=2}^{K}t_{k}^{*}(t_{k}^{*}x_{k}+\sum_{l\neq k}t_{l}^{*}z_{kl}),\   (1-t_{1}^{*})\bar{y}=\sum_{k=2}^{K}t_{k}^{*}y_{k}$. We have ${\bar{x}}/{\bar{y}}\leq \max_{2\leq k\leq K}{x_{k}}/{y_{k}}$ since $\min_{1\leq k\leq K} x_{k}>max_{k\neq l}z_{kl}$ . Let $z^{+}=max_{k\neq l}z_{kl}$, then we have
		\begin{align*}
		f(1,0,...,0)-\max_{t_{1}\leq 1-t}f(t_{1},...,t_{K})&\geq \frac{x_{1}}{y_{1}}-\frac{t_{1}^{*}(t_{1}^{*}x_{1}+(1-t_{1}^{*})z^{+})+(1-t_{1}^{*})\bar{x}}{t_{1}^{*}y_{1}+(1-t_{1}^{*})\bar{y}}\\
		&=\frac{(1-t_{1}^{*})t_{1}^{*}y_{1}(x_{1}-z^{+})}{y_{1}(t_{1}^{*}y_{1}+(1-t_{1}^{*})\bar{y})}+\frac{(1-t_{1}^{*})(x_{1}\bar{y}-\bar{x}y_{1})}{y_{1}(t_{1}^{*}y_{1}+(1-t_{1}^{*})\bar{y})}.
		\end{align*}
		Case I: $y_{1} = \min_{1\leq k\leq K}y_{k}$.
		
		We have ${\bar{y}}/\left[t_{1}^{*}y_{1}+(1-t_{1}^{*})\bar{y}\right]\geq 1$, and thus	
		\begin{align*}
		\frac{(1-t_{1}^{*})(x_{1}\bar{y}-\bar{x}y_{1})}{y_{1}(t_{1}^{*}y_{1}+(1-t_{1}^{*})\bar{y})}&=(1-t_{1}^{*})(\frac{x_{1}}{y_{1}}-\frac{\bar{x}}{\bar{y}})\frac{\bar{y}}{t_{1}^{*}y_{1}+(1-t_{1}^{*})\bar{y}}\\
		&\geq t(\frac{x_{1}}{y_{1}}-\max_{2\leq k\leq K}\frac{x_{k}}{y_{k}}).
		\end{align*}
		Case II: $y_{1} > \min_{1\leq k\leq K}y_{k}$.
		
		There exists $i\neq 1$ such that $y_{i}<y_{1}$, then we have ${z^{+}}/{y_{1}}<{x_{i}}/{y_{i}}$.
		
		If $\bar{y}(1-t_{1}^{*})\geq y_{1}t_{1}^{*}$, we have
		$$\frac{\bar{y}}{t_{1}^{*}y_{1}+(1-t_{1}^{*})\bar{y}}\geq \frac{1}{2(1-t_{1}^{*})},$$
		and
		\begin{align*}
		\frac{(1-t_{1}^{*})(x_{1}\bar{y}-\bar{x}y_{1})}{y_{1}(t_{1}^{*}y_{1}+(1-t_{1}^{*})\bar{y})}&=(1-t_{1}^{*})(\frac{x_{1}}{y_{1}}-\frac{\bar{x}}{\bar{y}})\frac{\bar{y}}{t_{1}^{*}y_{1}+(1-t_{1}^{*})\bar{y}}\\
		&\geq \frac{1}{2}(\frac{x_{1}}{y_{1}}-\max_{2\leq k\leq K}\frac{x_{k}}{y_{k}}).
		\end{align*}
		If $\bar{y}(1-t_{1}^{*})\leq y_{1}t_{1}^{*}$, we have
		$$\frac{t_{1}^{*}y_{1}}{t_{1}^{*}y_{1}+(1-t_{1}^{*})\bar{y}}\geq \frac{1}{2}$$ and ${z^{+}}/{y_{1}}\leq \max_{2\leq k\leq K}{x_{k}}/{y_{k}}$, then we have
		\begin{align*}\frac{(1-t_{1}^{*})t_{1}^{*}y_{1}(x_{1}-z^{+})}{y_{1}(t_{1}^{*}y_{1}+(1-t_{1}^{*})\bar{y})}&=(1-t_{1}^{*})(\frac{x_{1}}{y_{1}}-\frac{z^{+}}{y_{1}})\frac{t_{1}^{*}y_{1}}{t_{1}^{*}y_{1}+(1-t_{1}^{*})\bar{y}}\\
		&\geq \frac{1}{2}t(\frac{x_{1}}{y_{1}}-\max_{2\leq k\leq K}\frac{x_{k}}{y_{k}}).
		\end{align*}
		So we have $$f(1,0,...,0)-\max_{t_{1}\leq 1-t}f(t_{1},...,t_{K})\geq \frac{1}{2}t\left(\frac{x_{1}}{y_{1}}-\max_{2\leq k\leq K}\frac{x_{k}}{y_{k}}\right).$$
	\end{enumerate}
\end{proof}

Based on the lemmas given previously, we give the proof of Theorem \ref{DCSBM}.

\begin{proof} [Proof of Theorem \ref{DCSBM}]
	Based on Lemma \ref{cd}, with probability at least $1-2^{n+2}/n^n$, we have
	$$\max_{S\in \Gamma_{\delta}}\left|\psi(S)-\hat{\psi}(S)\right|\lesssim \frac{1}{n^{\alpha-\delta}}.$$
	Let $t_{k}(S)={r_{k}^{d}(S)}/{r^{d}(S)}$ for $1\leq k\leq K$,
	\begin{align*}
	\hat{\psi}(S)&=\frac{\sum_{k=1}^{K}nr_{k}^{d}(S)(\sum_{l=1}^{K}nr_{l}^{d}(S)p_{kl})}{\sum_{k=1}^{K}nr_{k}^{d}(S)(\sum_{l=1}^{K}n\hat{\pi}_{l}^{d}p_{kl})}-n\eta\sum_{k=1}^{K}\frac{\pi_{k}}{\pi_{k}^{d}}r_{k}^{d}(S)\\
	&=r^{d}(S)\left(\frac{\sum_{k=1}^{K}t_{k}(S)(t_k(S)p_{kk}+\sum_{l\neq k}t_{l}(S)p_{kl})}{\sum_{k=1}^{K}t_{k}(S)(\sum_{l=1}^K\hat{\pi}_{l}^{d}p_{kl})}-n\eta\sum_{k=1}^{K}\frac{\pi_{k}}{\pi_{k}^{d}}t_{k}(S)\right)\\
	&=r^{d}(S)\left(f(t_{1}(S),...,t_{K}(S))-n\eta\sum_{k=1}^{K}\frac{\pi_{k}}{\pi_{k}^{d}}t_{k}(S)\right).
	\end{align*}
Note that if $S=G_1$, $t_1(G_1)=1$, $t_k(G_1)=0$ ($k=2,\cdots,K$) and
$$\hat{\psi}(G_1)=\hat{\pi}_1^d\left(f(1,0,\cdots,0)-n\eta\frac{\pi_{1}}{\pi_{1}^{d}}\right).$$	
	Based on Lemma \ref{rho}, with probability at least $1-2Kn^{-2}$ we have
	$$\hat{\rho}_{1}^{d}-\max_{2\leq k\leq K}\hat{\rho}_{k}^{d}\gtrsim \frac{1}{n^{\tau}}.$$
	Then, with probability at least $1-2Kn^{-2}$ we have
	\begin{align}\label{etanoempty}&f(1,0,...,0)-\max_{t_{1}\leq 1-1/n^{\gamma-\tau}}f(t_{1},t_{2},...,t_{K})\nonumber \\
	&\gtrsim \frac{1}{n^{\gamma-\tau}}(\hat{\rho}_{1}^{d}-\max_{2\leq k\leq K}\hat{\rho}_{k}^{d})
	\gtrsim \frac{1}{n^{\gamma-\tau}}\frac{1}{n^{\tau}}
	\gtrsim \frac{1}{n^{\gamma}},
	\end{align}
	and
	$f(t_{1},...,t_{K})\leq f(1,0,...,0)=\hat{\rho}_{1}^{d}$
	by Lemma \ref{Ceta}. From (\ref{etanoempty}), it is easy to see that  for a constant $C$, we can choose $\eta$ satisfying the inequality (\ref{etacondition}) with probability at least $1-2Kn^{-2}$.
	
Since $\max_{2\leq k\leq K}{\pi_{k}^{d}}/{\pi_{k}}\leq{\pi_{1}^{d}}/{\pi_{1}}\leq h_{M}$,
	Using the inequality in the condition, we have with sufficiently large $n$,
	\begin{equation*}
	f(t_{1}(S),...,t_{K}(S))-n\eta\sum_{k=1}^{K}\frac{\pi_{k}}{\pi_{k}^{d}}t_{k}(S)\begin{cases}
	>{C}/{(h_{M}n^{\gamma})}, & \mbox{if}~\ t_{1}(S)=1;\\
	<-{C}/{(h_{M}n^{\gamma})},& \mbox{if}~ t_{1}(S)\leq 1-1/n^{\gamma-\tau}.
	\end{cases}
	\end{equation*}
	So we have
	\begin{equation*}
	\hat{\psi}(S)\begin{cases}
	>{r^{d}(S) C}/{(h_{M}n^{\gamma})}, & \mbox{if}~\ t_{1}=1;\\
	<-{r^{d}(S) C}/{(h_{M}n^{\gamma})},& \mbox{if}~\ t_{1}\leq 1-1/n^{\gamma-\tau}.
	\end{cases}
	\end{equation*}
	with sufficiently large $n$.
	To maximize $\hat{\psi}(S)$, $t_{1}$ must be bigger than $1-1/n^{\gamma-\tau}$. If $r^{d}(S)>\hat{\pi}_{1}^{d}+{\hat{\pi}_{1}^{d}}/{(n^{\gamma-\tau}-1)}$, then $t_{1}\leq {\hat{\pi}_{1}^{d}}/{r^{d}(S)}<1-1/n^{\gamma-\tau}$. So we must have $r^{d}(S)\leq\hat{\pi}_{1}^{d}+{\hat{\pi}_{1}^{d}}/{(n^{\gamma-\tau}-1)}$. If $r^{d}(S)\leq \hat{\pi}_{1}^{d}-{\hat{\pi}_{1}^{d}\log n}/{n^{\alpha-2\delta-\gamma}}$ and $\hat{\pi}_{1}^{d}\gtrsim {n^{-\delta}}$, by Lemma \ref{Ceta} we have
	\begin{align*}
	\hat{\psi}(G_{1})-\hat{\psi}(S)&\geq (\hat{\pi}_{1}^{d}-r^{d}(S))\left(f(1,0,...,0)-n\eta\frac{\pi_{1}}{\pi_{1}^{d}}\right)\\
	&\geq \hat{\pi}_{1}^{d}\frac{\log n}{n^{\alpha-2\delta-\gamma}}\left(f(1,0,...,0)-n\eta\frac{\pi_{1}}{\pi_{1}^{d}}\right)
	\gtrsim \frac{\log n}{n^{\alpha-\delta}}.
	\end{align*}
Since $r^{d}(S)=\sum_{k=1}^{K}r_{k}(S){\pi_{k}^{d}}/{\pi_{k}}\gtrsim {1}/{n^{\delta}}$ and $\pi_{1}^{d}\gtrsim {1}/{n^{\delta}}$, then with probability at least $1-2Kn^{-2}$ we have
	$$\hat{\psi}(G_{1})-\hat{\psi}(S)
	\begin{cases}\gtrsim \frac{\log n}{n^{\alpha-\delta}}, & \mbox{if} \ ~r^{d}(S)\leq \hat{\pi}_{1}^{d}-\frac{\hat{\pi}_{1}^{d}\log n}{n^{\alpha-2\delta-\gamma}}\\
	\gtrsim \frac{C}{n^{\gamma+\delta}}, & \mbox{if}~\ t_{1}(S)\leq 1-1/n^{\gamma-\tau}
	\end{cases}$$
	
	So with probability at least $1-2Kn^{-2}-2^{n+2}/n^n$, $\psi(S)<\psi(G_{1})$ for all $S$ satisfying $r^{d}(S)\leq \hat{\pi}_{1}^{d}-{\hat{\pi}_{1}^{d}\log n}/{n^{\alpha-2\delta-\gamma}}$ or $t_{1}(S)\leq 1-1/n^{\gamma-\tau}$, which implies that $\psi(S)$ maximizes when
	$t_{1}(S)\geq 1-1/n^{\gamma-\tau} \ \ and \ \  \hat{\pi}_{1}^{d}-{\hat{\pi}_{1}^{d}\log n}/{n^{\alpha-2\delta-\gamma}}\leq r^{d}(S)\leq\hat{\pi}_{1}^{d}+{\hat{\pi}_{1}^{d}}/{(n^{\gamma-\tau}-1)}$
	with probability at least $1-2Kn^{-2}-2^{n+2}/n^n$. Since $t_1(S)=r^d_1(S)/r^d(S)$, we therefore have $\hat{\pi}_{1}^{d}-{\hat{\pi}_{1}^{d}\log n}/{n^{\alpha-2\delta-\gamma}}\leq {r_{1}^{d}(S)}/{t_{1}(S)}\leq\hat{\pi}_{1}^{d}+{\hat{\pi}_{1}^{d}}/{(n^{\gamma-\tau}-1)}$, and hence $(1-1/n^{\gamma-\tau})\hat{\pi}_{1}\left(1-{\log n}/{n^{\alpha-2\delta-\gamma}}\right)\leq r_{1}(S)\leq\hat{\pi}_{1}t_1(S)n^{\gamma-\tau}/{(n^{\gamma-\tau}-1)}$. From $t_{1}(S)\leq 1-1/n^{\gamma-\tau}$, we get $h_1h_M^{-1}(r(S)-r_1(S))/r(S)\leq(r^d(S)-r_1^d(S))/r^d(S)\leq 1/n^{\gamma-\tau}$ and $r_1(S)/r(S)\geq 1-h_{M}h_{1}^{-1}/n^{\gamma-\tau}$. Note that $r_1(S)/\hat{\pi}_1=|S\bigcap G_1|/|G_1|$ and $r_1(S)/r(S)=|S\bigcap G_1|/|S|$. Therefore, with probability at least $1-2Kn^{-2}-2^{n+2}/n^n$, we have
	\begin{equation}
	\frac{\left|S\Delta G_{1}\right|}{\left|S\bigcup G_{1}\right|}\leq 2h_{M}h_{1}^{-1}/n^{\gamma-\tau}+\log n/n^{\alpha-2\delta-\gamma}.
	\end{equation}
\end{proof}

The proof of Theorem \ref{DCSBM_outlier} is very similar to the proof of Theorem \ref{DCSBM} and we omit it. The only difference is that we have to pay attention to the ourlier community. For example, in the proof of Lemma \ref{cd}, the inequality (\ref{WS_ExpectationIneq}) becomes $\mathbb{E}(W(S)|\mathbf{c})\gtrsim n^{1+2\alpha}\log n \sum_{k=1}^{K-1}(r_{k}^{d}(S))^{2}$. Then, using the condition  $|G_K|^2/K=o(n^{2-2\delta})$, we can also get $\mathbb{E}(W(S)|\mathbf{c}) \gtrsim n^{1+2(\alpha-\delta)}\log n$ and hence the conclusion of Lemma \ref{cd} also holds for DCSBM with outliers.



\begin{thebibliography}{100}

\bibitem[\protect\citeauthoryear{Amini, Chen, Bickel, and Levina}{Amini
  et~al.}{2013}]{amini2013pseudo}
Amini, A.~A., A.~Chen, P.~J. Bickel, and E.~Levina (2013).
\newblock Pseudo-likelihood methods for community detection in large sparse
  networks.
\newblock {\em The Annals of Statistics\/}~{\em 41\/}(4), 2097--2122.

\bibitem[\protect\citeauthoryear{Balakrishnan, Xu, Krishnamurthy, and
  Singh}{Balakrishnan et~al.}{2011}]{balakrishnan2011noise}
Balakrishnan, S., M.~Xu, A.~Krishnamurthy, and A.~Singh (2011).
\newblock Noise thresholds for spectral clustering.
\newblock {\em In Advances in Neural Information Processing Systems\/},
  954--962.

\bibitem[\protect\citeauthoryear{Bickel and Chen}{Bickel and
  Chen}{2009}]{bickel2009nonparametric}
Bickel, P.~J. and A.~Chen (2009).
\newblock A nonparametric view of network models and {N}ewman--{G}irvan and
  other modularities.
\newblock {\em Proceedings of the National Academy of Sciences\/}~{\em
  106\/}(50), 21068--21073.

\bibitem[\protect\citeauthoryear{Bickel, Choi, Chang, and Zhang}{Bickel
  et~al.}{2013}]{bickel2013asymptotic}
Bickel, P.~J., D.~Choi, X.~Chang, and H.~Zhang (2013).
\newblock Asymptotic normality of maximum likelihood and its variational
  approximation for stochastic blockmodels.
\newblock {\em The Annals of Statistics\/}~{\em 41\/}(4), 1922--1943.

\bibitem[\protect\citeauthoryear{Cai and Li}{Cai and Li}{2015}]{cai2015robust}
Cai, T. and X.~Li (2015).
\newblock Robust and computationally feasible community detection in the
  presence of arbitrary outlier nodes.
\newblock {\em The Annals of Statistics\/}~{\em 43\/}(3), 1027--1059.

\bibitem[\protect\citeauthoryear{Chaudhuri, Graham, and Tsiatas}{Chaudhuri
  et~al.}{2012}]{chaudhuri2012spectral}
Chaudhuri, K., F.~C. Graham, and A.~Tsiatas (2012).
\newblock Spectral clustering of graphs with general degrees in the extended
  planted partition model.
\newblock {\em Journal of Machine Learning Research\/}~{\em 35}, 1--23.

\bibitem[\protect\citeauthoryear{Choi, Wolfe, and Airoldi}{Choi
  et~al.}{2012}]{choi2012stochastic}
Choi, D., P.~Wolfe, and E.~Airoldi (2012).
\newblock Stochastic blockmodels with a growing number of classes.
\newblock {\em Biometrika\/}~{\em 99\/}(2), 273--284.

\bibitem[\protect\citeauthoryear{Chung}{Chung}{1997}]{chung1997spectral}
Chung, F.~R. (1997).
\newblock {\em Spectral Graph Theory}, Volume~92.
\newblock American Mathematical Soc.

\bibitem[\protect\citeauthoryear{Decelle, Krzakala, Moore, and
  Zdeborov{\'a}}{Decelle et~al.}{2011}]{decelle2011asymptotic}
Decelle, A., F.~Krzakala, C.~Moore, and L.~Zdeborov{\'a} (2011).
\newblock Asymptotic analysis of the stochastic block model for modular
  networks and its algorithmic applications.
\newblock {\em Physical Review E\/}~{\em 84\/}(6), 66--106.

\bibitem[\protect\citeauthoryear{Fortunato and Barth{\'e}lemy}{Fortunato and
  Barth{\'e}lemy}{2007}]{fortunato2007resolution}
Fortunato, S. and M.~Barth{\'e}lemy (2007).
\newblock Resolution limit in community detection.
\newblock {\em Proceedings of the National Academy of Sciences\/}~{\em
  104\/}(1), 36--41.

\bibitem[\protect\citeauthoryear{F{\"u}redi and Koml{\'o}s}{F{\"u}redi and
  Koml{\'o}s}{1981}]{furedi1981eigenvalues}
F{\"u}redi, Z. and J.~Koml{\'o}s (1981).
\newblock The eigenvalues of random symmetric matrices.
\newblock {\em Combinatorica\/}~{\em 1\/}(3), 233--241.

\bibitem[\protect\citeauthoryear{Hagen and Kahng}{Hagen and
  Kahng}{1992}]{hagen1992new}
Hagen, L. and A.~B. Kahng (1992).
\newblock New spectral methods for ratio cut partitioning and clustering.
\newblock {\em IEEE Transactions on Computer-aided Design of Integrated
  Circuits and Systems\/}~{\em 11\/}(9), 1074--1085.

\bibitem[\protect\citeauthoryear{Hoeffding}{Hoeffding}{1963}]{hoeffding1963probability}
Hoeffding, W. (1963).
\newblock Probability inequalities for sums of bounded random variables.
\newblock {\em Journal of the American Statistical Association\/}~{\em
  58\/}(301), 13--30.

\bibitem[\protect\citeauthoryear{Holland, Laskey, and Leinhardt}{Holland
  et~al.}{1983}]{holland1983stochastic}
Holland, P., K.~Laskey, and S.~Leinhardt (1983).
\newblock Stochastic blockmodels: First steps.
\newblock {\em Social Networks\/}~{\em 5\/}(2), 109--137.

\bibitem[\protect\citeauthoryear{Jin}{Jin}{2015}]{jin2015fast}
Jin, J. (2015).
\newblock Fast community detection by {SCORE}.
\newblock {\em The Annals of Statistics\/}~{\em 43\/}(1), 57--89.

\bibitem[\protect\citeauthoryear{Joseph and Yu}{Joseph and
  Yu}{2016}]{joseph2016impact}
Joseph, A. and B.~Yu (2016).
\newblock Impact of regularization on spectral clustering.
\newblock {\em The Annals of Statistics\/}~{\em 44\/}(4), 1765--1791.

\bibitem[\protect\citeauthoryear{Karrer and Newman}{Karrer and
  Newman}{2011}]{karrer2011stochastic}
Karrer, B. and M.~E.~J. Newman (2011).
\newblock Stochastic blockmodels and community structure in networks.
\newblock {\em Physical Review E\/}~{\em 83\/}(1), 16--107.

\bibitem[\protect\citeauthoryear{Khorasgani, Chen, and Zaiane}{Khorasgani
  et~al.}{2010}]{khorasgani2010top}
Khorasgani, R., J.~Chen, and O.~Zaiane (2010).
\newblock Top leaders community detection approach in information networks.
\newblock In {\em 4th SNA-KDD workshop on social network mining and analysis}.

\bibitem[\protect\citeauthoryear{Kim and Shi}{Kim and
  Shi}{2012}]{kim2012scalable}
Kim, S. and T.~Shi (2012).
\newblock Scalable spectral algorithms for community detection in directed
  networks.
\newblock {\em arXiv preprint arXiv:1211.6807\/}.

\bibitem[\protect\citeauthoryear{Kumar, Novak, and Tomkins}{Kumar
  et~al.}{2010}]{kumar2010structure}
Kumar, R., J.~Novak, and A.~Tomkins (2010).
\newblock Structure and evolution of online social networks.
\newblock {\em In Link Mining: Models, Algorithms, and Applications\/},
  337--357.

\bibitem[\protect\citeauthoryear{Lancichinetti, Radicchi, Ramasco, and
  Fortunato}{Lancichinetti et~al.}{2011}]{lancichinetti2011finding}
Lancichinetti, A., F.~Radicchi, J.~J. Ramasco, and S.~Fortunato (2011).
\newblock Finding statistically significant communities in networks.
\newblock {\em PLOS ONE\/}~{\em 6\/}(4), e18961.

\bibitem[\protect\citeauthoryear{Lei and Rinaldo}{Lei and
  Rinaldo}{2015}]{lei2015consistency}
Lei, J. and A.~Rinaldo (2015).
\newblock Consistency of spectral clustering in stochastic block models.
\newblock {\em The Annals of Statistics\/}~{\em 43\/}(1), 215--237.

\bibitem[\protect\citeauthoryear{Leskovec, Lang, Dasgupta, and
  Mahoney}{Leskovec et~al.}{2008}]{leskovec2008statistical}
Leskovec, J., K.~J. Lang, A.~Dasgupta, and M.~W. Mahoney (2008).
\newblock Statistical properties of community structure in large social and
  information networks.
\newblock {\em In Proceedings of the 17th international conference on World
  Wide Web\/}, 695--704.

\bibitem[\protect\citeauthoryear{Mariadassou, Robin, and Vacher}{Mariadassou
  et~al.}{2010}]{mariadassou2010uncovering}
Mariadassou, M., S.~Robin, and C.~Vacher (2010).
\newblock Uncovering latent structure in valued graphs: a variational approach.
\newblock {\em The Annals of Applied Statistics\/}~{\em 4\/}(2), 715--742.

\bibitem[\protect\citeauthoryear{Newman}{Newman}{2004a}]{newman2004coauthorship}
Newman, M. E.~J. (2004a).
\newblock Coauthorship networks and patterns of scientific collaboration.
\newblock {\em Proceedings of the National Academy of Sciences\/}~{\em
  101\/}(suppl 1), 5200--5205.

\bibitem[\protect\citeauthoryear{Newman}{Newman}{2004b}]{newman2004fast}
Newman, M. E.~J. (2004b).
\newblock Fast algorithm for detecting community structure in networks.
\newblock {\em Physical Review E\/}~{\em 69\/}(6), 66--133.

\bibitem[\protect\citeauthoryear{Newman}{Newman}{2006}]{newman2006modularity}
Newman, M. E.~J. (2006).
\newblock Modularity and community structure in networks.
\newblock {\em Proceedings of the National Academy of Sciences\/}~{\em
  103\/}(23), 8577--8582.

\bibitem[\protect\citeauthoryear{Newman and Girvan}{Newman and
  Girvan}{2004}]{newman2004finding}
Newman, M. E.~J. and M.~Girvan (2004).
\newblock Finding and evaluating community structure in networks.
\newblock {\em Physical Review E\/}~{\em 69\/}(2), 26--113.

\bibitem[\protect\citeauthoryear{Nowicki and Snijders}{Nowicki and
  Snijders}{2001}]{nowicki2001estimation}
Nowicki, K. and T.~Snijders (2001).
\newblock Estimation and prediction for stochastic blockstructures.
\newblock {\em Journal of the American Statistical Association\/}~{\em
  96\/}(455), 1077--1087.

\bibitem[\protect\citeauthoryear{Rohe, Chatterjee, and Yu}{Rohe
  et~al.}{2011}]{rohe2011spectral}
Rohe, K., S.~Chatterjee, and B.~Yu (2011).
\newblock Spectral clustering and the high-dimensional stochastic blockmodel.
\newblock {\em The Annals of Statistics\/}~{\em 39\/}(4), 1878--1915.

\bibitem[\protect\citeauthoryear{Shi and Malik}{Shi and
  Malik}{2000}]{shi2000normalized}
Shi, J. and J.~Malik (2000).
\newblock Normalized cuts and image segmentation.
\newblock {\em IEEE Transactions on Pattern Analysis and Machine
  Intelligence\/}~{\em 22\/}(8), 888--905.

\bibitem[\protect\citeauthoryear{Wang, Wang, Yu, and Zhang}{Wang
  et~al.}{2015}]{wang2015community}
Wang, M., C.~Wang, J.~X. Yu, and J.~Zhang (2015).
\newblock Community detection in social networks: an in-depth benchmarking
  study with a procedure-oriented framework.
\newblock {\em Proceedings of the VLDB Endowment\/}~{\em 8\/}(10), 998--1009.

\bibitem[\protect\citeauthoryear{Wei and Cheng}{Wei and
  Cheng}{1989}]{wei1989towards}
Wei, Y.~C. and C.~K. Cheng (1989).
\newblock Towards efficient hierarchical designs by ratio cut partitioning.
\newblock {\em In 1989 IEEE International Conference on Computer-Aided Design.
  Digest of Technical Papers\/}, 298--301.

\bibitem[\protect\citeauthoryear{White and Smyth}{White and
  Smyth}{2005}]{white2005spectral}
White, S. and P.~Smyth (2005).
\newblock A spectral clustering approach to finding communities in graph.
\newblock {\em In SDM\/}, 76--84.

\bibitem[\protect\citeauthoryear{Xu, Yuruk, Feng, and Schweiger}{Xu
  et~al.}{2007}]{xu2007scan}
Xu, X., N.~Yuruk, Z.~Feng, and T.~A. Schweiger (2007).
\newblock Scan: a structural clustering algorithm for networks.
\newblock {\em In Proceedings of the 13th ACM SIGKDD international conference
  on Knowledge discovery and data mining\/}, 824--833.

\bibitem[\protect\citeauthoryear{Yao}{Yao}{2003}]{yao2003information}
Yao, Y.~Y. (2003).
\newblock Information-theoretic measures for knowledge discovery and data
  mining.
\newblock {\em In Entropy Measures, Maximum Entropy Principle and Emerging
  Applications\/}, 115--136.

\bibitem[\protect\citeauthoryear{Yu et~al.}{Yu et~al.}{2008}]{yu2008high}
Yu, H. et~al. (2008).
\newblock High-quality binary protein interaction map of the yeast interactome
  network.
\newblock {\em Science\/}~{\em 322\/}(5898), 104--110.

\bibitem[\protect\citeauthoryear{Zhao, Levina, and Zhu}{Zhao
  et~al.}{2011}]{zhao2011community}
Zhao, Y., E.~Levina, and J.~Zhu (2011).
\newblock Community extraction for social networks.
\newblock {\em Proceedings of the National Academy of Sciences\/}~{\em
  108\/}(18), 7321--7326.

\bibitem[\protect\citeauthoryear{Zhao, Levina, and Zhu}{Zhao
  et~al.}{2012}]{zhao2012consistency}
Zhao, Y., E.~Levina, and J.~Zhu (2012).
\newblock Consistency of community detection in networks under degree-corrected
  stochastic block models.
\newblock {\em The Annals of Statistics\/}~{\em 40\/}(4), 2266--2292.

\end{thebibliography}
\end{document}